\documentclass[11pt,a4paper]{article} 
\usepackage{a4wide} 
\usepackage[utf8]{inputenc}
\usepackage[OT2,T1]{fontenc}
\usepackage[english]{babel}
\usepackage{pdfsync}
\usepackage{amsmath,amsthm,amssymb,cmll,bussproofs}
\usepackage{tikz}
\usepackage{hyperref}
\usepackage{enumitem}
\usepackage{color}
\usepackage{cleveref}
\usepackage{aliascnt}

\theoremstyle{plain}
\newtheorem{theorem}{Theorem}
\newtheorem{corollary}{Corollary}[theorem]
\newaliascnt{claim}{theorem}
\newtheorem{claim}[claim]{Claim}
\aliascntresetthe{claim}
     \crefname{claim}{claim}{claims}
\newaliascnt{proposition}{theorem}
\newtheorem{proposition}[proposition]{Proposition}
\aliascntresetthe{proposition}
     \crefname{proposition}{proposition}{propositions}
\newaliascnt{lemma}{theorem}
\newtheorem{lemma}[lemma]{Lemma}
\aliascntresetthe{lemma}
 \crefname{lemma}{lemma}{lemmas}

\theoremstyle{definition}
\newtheorem{definition}{Definition}

\theoremstyle{remark}
\newtheorem{remark}{Remark}

\title{Agafonov's Proof of Agafonov's Theorem: A Modern Account and New Insights}
\author{Thomas Seiller and Jakob Grue Simonsen}
\date{}

\newcommand{\seq}[1]{\mathbf{#1}}
\newcommand{\countones}[2][1]{\#_{#1}(#2)}
\newcommand{\card}{\mathrm{Card}}
\newcommand{\fqcy}[3][x]{{\mathrm{freq}_{#2}(#1;#3)}}
\newcommand{\naturalN}{\mathbf{N}}
\newcommand{\abs}[1]{\mathopen{|}#1\mathclose{|}}
\newcommand{\len}[1]{\mathrm{len}(#1)}
\newcommand{\pickedout}[2]{#2[\seq{#1}]}
\newcommand{\freq}[2][1]{\mathrm{freq}_{#1}(#2)}

\newcommand{\enquote}[1]{``#1''}

\begin{document}
\maketitle


\begin{abstract}
We give a modern account of Agafonov's original proof of his eponymous theorem. The original proof was only reported in Russian \cite{Agafonov,AgafonovRussianLong} in a journal not widely available,
and the work most commonly cited in western literature is instead the English translation \cite{Agafonovsummary} of a summary version containing no proofs \cite{AgafonovsummaryRussian}, and the main proof relied heavily on material well-known in Russian mathematical circles of the day, which perhaps obscures the main thrust of argumentation for modern readers.

Our present account recasts Aganofov's arguments using more basic building blocks than in the original proof, and contains some further embellishments to Agafonov’s original arguments, made in the interest of clarity. We posit that the modern account provides new insight to the underlying phenomena of the theorem.

We also provides some historical context to Agafonov's work, including a short description of some of the ideas that led to Agafonov's own proof, especially emphasizing the important work of Postnikova.
\end{abstract}

We give an account of Agafonov's original proof of his eponymous theorem. The original proof was only reported in Russian \cite{Agafonov,AgafonovRussianLong} in a journal not widely available,
and the work most commonly cited in western literature is instead the english translation \cite{Agafonovsummary} of a summary version containing no proofs \cite{AgafonovsummaryRussian}.

The account contains some embellishments to Agafonov's original arguments, made in the interest of clarity:

\begin{enumerate}


\item The original proof relies on results of Postnikova \cite{Postnikova}. We detail Postnikova's contribution and provide some historic context to her result.

\item The original proof contained a mixture of arguments expressed both via running text and explicit lemmas and theorems. While we have retained the general flow
of argumentation from the original, we have used explicit lemmas and propositions for a number of observations occurring in the running text.

\item We have made several arguments explicit and provided detailed arguments in places where Agafonov relied on immediate understanding
from his specialist audience, but where we believe that non-expert readers with modern sensibilities might prefer more elaborate explanations. 
The most pertinent examples are:
\begin{enumerate}

\item  We explicitly prove why it suffices to prove that a connected finite automaton picks out $b \in \{0,1\}^n$ for $n=1$ with limiting frequency $p$
from any $p$-distributed sequence (\Cref{lem:simply_normal_is_enough}).

\item We have appealed directly to probabilistic reasoning (using Chebyshev's Inequality) in the proof that, \emph{par abus de langage}, the probability of deviation
from probability $p$ among the symbols selected by a finite automaton from sets of substrings picked from a $p$-distributed sequence tends to zero with increasing length of the strings 
(\Cref{main:lemma2}). In \cite{AgafonovRussianLong}, this was essentially proved by a reference to the Strong Law of Large Numbers and a statement that the proof was similar to 
Lemma 3 of Loveland's paper \cite{Loveland66}.

\end{enumerate}

\end{enumerate}

\paragraph{Acknowledgements.} The authors warmly thank \href{https://www.mimuw.edu.pl/~lukaszcz/}{Łukasz Czajka} and \href{https://avolkova.org}{Anastasia Volkova} for their help in translating the russian documents.


\section{Preliminaries} 

If $\alpha = a_1 a_1 \cdots $ is a right-infinite sequence
over an alphabet $\mathcal{A}$ and $N$ is a positive integer, we denote by
$\alpha \vert_{\leq N}$ the finite string $a_1 a_2 \cdots a_N$.

We denote by $\mathcal{A}^*$ the set of (finite) words over $\mathcal{A}$ and by $\mathcal{A}^+$ the set of finite \emph{non-empty} words over $\mathcal{A}$.

\begin{definition}
A finite probability map (over an alphabet $\mathcal{A}$) is a map $p : \mathcal{A}^+ \longrightarrow [0,1]$
such that, for all positive integers $n$, $\sum_{a_1 \cdots a_n \in \mathcal{A}^n} p(a_1 \cdots a_n) = 1$.

A finite probability map $p$ is said to be:

\begin{itemize}

\item  \emph{Bernoulli} if, for all positive integers $n$,
and all $a_1,\ldots,a_n \in \mathcal{A}$,
$p(a_1 \cdots a_n) = \prod_{j=1}^n p(a_j)$.

\item \emph{Equidistributed} if, for any string
$a_1 \cdots a_n \in \mathcal{A}^{n}$,
 $p(a_1 \cdots a_n) = \vert \mathcal{A} \vert^{-n}$.

\end{itemize}
\end{definition}

Observe that an equidistributed $p$ is also Bernoulli.
For alphabets $\vert \mathcal{A} \vert > 1$,
any map $g : \mathcal{A} \longrightarrow
[0,1]$ with $\sum_{a \in \mathcal{A}} g(a) = 1$
induces a Bernoulli finite probability map $p_g$ by
letting $p_g(a_1 \cdots a_n) \triangleq \prod_{j=1}^n g(a_j)$.
This map is equidistributed if{f} $g(a) = \vert A \vert^{-1}$
for every $a \in \mathcal{A}$.

The use of the word ``Bernoulli'' is due to the fact that Bernoulli finite probability maps
correspond directly to the measure of cylinders in Bernoulli shifts \cite{shields:bernoulli}; in the literature
on normal numbers, the word Bernoulli is sometimes used slightly differently, for example Schnorr and Stimm \cite{DBLP:journals/acta/SchnorrS72}
use the term Bernoulli sequences for sequences distributed according to 
finite probability map that are equidistributed in our terminology.

We are interested in the finite probability maps whose values can be realized as the limiting frequencies 
of finite words in right-infinite sequences over $\{0,1\}$.

\begin{definition}\label{def:pdistrib}
Let ${b} = b_1 \cdots b_N$ and ${a} = a_1 \cdots a_n$ be finite words over $\mathcal{A}$. We denote by
$\countones[{a}]{{b}}$ the number of occurrences of ${a}$ in ${b}$, that is, the quantity
$$
\left\vert \left\{j : b_j b_{j+1} \cdots b_{j+n-1} = a_1 a_2 \cdots a_n \right\}\right\vert
$$

Let $p$ be a finite probability map over $\mathcal{A}$, and be $\alpha$ is a right-infinite sequence over $\mathcal{A}$.
If the limit
$$
\freq[{a}]{\alpha}=\lim_{N \rightarrow \infty} \frac{\countones[{{a}}]{\alpha\vert_{\leq_N}}}{N}
$$
exists and is equal to some real number $f$, we say that ${a}$ occurs in $\alpha$ \emph{with limiting frequency} $f$.
If every ${a} \in \mathcal{A}^+$ occurs in $\alpha$ with limiting frequency $p({a})$,
we say that $\alpha$ is $p$-distributed.
\end{definition}

Observe that a right-infinite sequence $\alpha$ is normal in the usual sense if{f} it is $p$-distributed for (the unique) equidistributed
finite probability map $p$ over $\mathcal{A}$. Also observe that it is not all finite probability maps $p$ for which there
exists a $p$-distributed sequence.

An example of a finite probability map that is \emph{not} Bernoulli,
but such that there is at least one $p$-distributed right-infinite sequence,
is the map $b$ defined by $b(\alpha) = 1/2$ if $\alpha$ does not contain either of the
strings $00$ or $11$ (note that for each positive integer $n$, there
are exactly two such strings of length $n$ of each length), and
$b(\alpha) = 0$ otherwise. Observe that the right-infinite sequence
$010101 \cdots$ is $p$-distributed.

In the remaining sections, we will work with the alphabet $\{0,1\}$ unless otherwise specified.

\section{Preliminaries and Historical aspects} 

\subsection{Borel}

The notion of $p$-distributed sequences can be traced back to a 1909 paper by \'{E}mile Borel \cite{Borel}. In this work, Borel studies the decimal representation of numbers and introduces the following definitions.

\begin{definition}[Borel normality]
Consider an integer $b>1$. Consider a number $0<a<1$ and denote by $\alpha^b$ its decimal sequence $a^b_1,\dots,a^b_n,\dots \in \{0,1,\dots,b-1\}^\omega$ in base $b$, i.e. $a=\sum_n \frac{a^b_n}{b^n}$. Then $x$ is said to be:
\begin{enumerate}
\item \emph{simply normal} w.r.t. the basis $b$ when $\freq[c]{\alpha}=\frac{1}{b}$ for all $c\in\{0,1,\dots,b-1\}$;
\item \emph{entirely normal} (or just \emph{normal}) w.r.t. the basis $b$ when for all integers $n,k$ the number $b^k x$ is simply normal w.r.t. the basis $b^m$;
\item \emph{absolutely normal} if it is entirely normal w.r.t. every possible basis $b$.
\end{enumerate}
\end{definition}

Borel already remarks that normality correspond to what we introduced as $p$-distribution\footnote{The translation is ours, in which we replaced the basis $10$ considered by Borel with a parametrised basis $b$.}:
\begin{quote}
The characterising property of a normal number\footnote{I.e. entirely normal w.r.t. the basis $b$, where $b=10$ in Borel's original paper.} is the following: considering a sequence of $p$ symbols, denoting by $c_n$ the number of times this sequence is to be found within the n first decimal numbers, we have $\lim_{n\rightarrow\infty} \frac{c_n}{n}=\frac{1}{b^p}$.
\end{quote}

The main result of Borel on normal numbers is the following theorem.

\begin{theorem}[Borel \cite{Borel}]
The probability that a number is absolutely normal is equal to 1, i.e. almost all numbers are absolutely normal.
\end{theorem}

As a consequence, the probability that a number is normal, or simply normal, is also equal to $1$. In particular, the cardinality of the set of normal numbers is equal to the cardinality of the continuum $2^{\aleph_0}$, and normal numbers are dense in the set of all real numbers.

\subsection{von Mises}

The notion of $p$-distributed sequences also appeared in connection with the notion of \emph{kollektiv} introduced by von Mises in order to capture the concept of \emph{random sequence}. The intuition behind von Mises approach it that a random sequence is one that cannot be predicted. I.e. the frequency of each possible outcome is independent from a the choice of a \emph{Spielsystem}, i.e. a way to predict the outcome of successive trials. In other words, a sequence of trials outcomes is \emph{not} random whenever there exists a strategy to select a subsequence of the trials in order to modify the frequency of the outcomes. This is expressed as the second condition in the following definition.
As reported by Church \cite{Church40}, a sequence $\alpha=a_1,a_2,\dots,a_n,\dots$ in $\{0,1\}^{\omega}$ is a \emph{kollektiv} according to von Mises \cite{vonMiseskollektiv,vonMises36} when:
\begin{enumerate}
\item $\freq[1]{\alpha}$ is defined and equal to $p$;
\item if $\beta=a_{n_1},a_{n_2},\dots$ is any infinite sub-sequence of $\alpha$ formed
by deleting some of the terms of the latter sequence according to a
rule which makes the deletion or retention of $a_n$ depend only on $n$ and
$a_1,a_2,\dots,a_{n-1}$, then $\freq[1]{\beta}$ is defined and equal to $p$.
\end{enumerate}
However, Church judges this definition to be "too inexact in form to serve satisfactorily as the basis of a mathematical theory" and proposes the following formalisation.

\begin{definition}[von Mises kollektiv]
Let $\alpha$ be a sequence $a_1,a_2,\dots,a_n,\dots$ in $\{0,1\}^{\omega}$. It is a \emph{kollektiv} (in the sense of von Mises, as formalised by Church) when:
\begin{enumerate}
\item $\freq[1]{\alpha}$ is defined and equal to $p$;
\item If $\varphi$ is any function of positive integers, if\footnote{Note that the terms $b_n$ are written as follows in binary $b_n=1a_1 a_2\dots a_{n-1}$.} $b_1=1$, $b_{n+1}=2b_n+a_n$, $c_n=\varphi(b_n)$, and the integers $n$ such that $c_n= 1$ form in order of magnitude an infinite sequence $n_1,n_2,\dots$, then the sequence $\beta=a_{n_1},a_{n_2},\dots$ satisfies that $\freq[1]{\beta}$ is defined and equal to $p$.
\end{enumerate}
\end{definition}

In this section, several other notions of kollektiv will be discussed and introduced, and we will therefore use the following definitions.

\begin{definition}[Strategy]
A strategy $S$ is a predicate over the set of finite binary words, i.e. $S\subset\{0,1\}^{*}=\cup_{i=0}^{\omega}\{0,1\}^{i}$.
\end{definition}

\begin{definition}[Selected Subsequence]
Given a strategy $S$ and an infinite sequence $\alpha=a_1,a_2,\dots,a_n,\dots$ in $\{0,1\}^{\omega}$, we define the sequence $S(\alpha)$ as follows. Let $i_{1},i_{2},\dots,i_{k},\dots$ be the (increasing) sequence of indices $j$ such that $\alpha\vert_{\leq j-1}\in S$.
\[ S(\alpha)_{j}= a_{i_{j}} \]
\end{definition}

\begin{definition}[Kollektiv]
A sequence $\alpha=a_1,a_2,\dots,a_n,\dots$ in $\{0,1\}^{\omega}$ is a \emph{kollektiv} w.r.t. a set of strategies $\mathbf{S}$ when:
\begin{enumerate}
\item $\freq[1]{\alpha}$ is defined and equal to $p$;
\item for any strategy $S\in\mathbf{S}$, $\freq[1]{S(\alpha)}$ is defined and equal to $p$.
\end{enumerate}
\end{definition}

\subsection{Church}

With this definition, the notion of von Mises kollektiv coincides with that of kollektiv w.r.t. the set of all strategies. As discussed by several authors \cite{Tornier29,Reichenbach32,Kamke33,Copeland36}, this notion of kollektiv is however inadequate, because it is too restrictive. This is further explained by Church, who explains why no kollektiv can exist if one considers such a strong notion:
\begin{quote}
[...] it makes the class of random sequences associated with any probability $p$
other than $0$ or $1$ an empty class. For the failure of (2) may always
be shown by taking $\varphi(x) = a_{\mu(x)}$ where $\mu(x)$ is the least positive integer $m$
such that $2^m >x$: the sequence $a_{n_1}, a_{n_2}, \dots$ will then consist of those and only 
those terms of $a_1, a_2,\dots$ which are 1's\footnote{Indeed, the function defined by Church ensures that $c_n=a_n$.}.
\end{quote} 
As a consequence, Church introduces a new notion of kollektiv, by factorising in the notion of computability. This choice is furtehr argumented as follows:
\begin{quote}
the scientist concerned with making predictions or probable predictions of some phenomenon must employ an effectively calculable function : if the law of the phenomenon is not approximable by such a function, prediction is impossible. Thus a Spielsystem should be represented mathematically, not as a function, or even as a definition of a function, but as an effective algorithm for the calculation of the values of a function.
\end{quote}

\begin{definition}[Church kollektiv]
Let $\alpha$ be a sequence $a_1,a_2,\dots,a_n,\dots$ in $\{0,1\}^{\omega}$. It is a \emph{kollektiv} (in the sense of Church) when:
\begin{enumerate}
\item $\freq[1]{\alpha}$ is defined and equal to $p$;
\item If $\varphi$ is any \emph{effectively calculable}\footnote{Today, one would rather use the terminology "computable".} function of positive integers, if $b_1=1$, $b_{n+1}=2b_n+a_n$, $c_n=\varphi(b_n)$, and the integers $n$ such that $c_n= 1$ form in order of magnitude an infinite sequence $n_1,n_2,\dots$, then the sequence $\beta=a_{n_1},a_{n_2},\dots$ satisfies that $\freq[1]{\beta}$ is defined and equal to $p$.
\end{enumerate}
\end{definition}

\subsection{Admissible sequences}

Towards the general purpose of defining mathematically the notion of random sequence, other notions were also considered at the time. For our purpose, the notions of "admissible number" introduced by Copeland \cite{Copeland28}, also studied by Reichenbach under the name "normal number" \cite{Reichenbach32,Reichenbach37} will be of interest. 

\begin{definition}[Copeland-admissible sequence.]
Let $\alpha=a_1,a_2,\dots,a_n,\dots$ be a sequence in $\{0,1\}^{\omega}$. For all integers $r,n$, define the sequence
\[ (r/n)\alpha = a_r, a_{r+n}, \dots, a_{r+kn}, \dots\]
The sequence $\alpha$ is \emph{admissible} (in the sense of Copeland) if the following are satisfied:
\begin{enumerate}
\item For all $r,n$, $\freq[1]{(r/n)\alpha}$ is defined and equal to $p$.
\item $(1/n)\alpha, (2/n)\alpha, \dots, (n/n)\alpha$ are \emph{independent}\footnote{Independence here is understood in terms of probability theory, as is detailed in Copeland's paper in which he states that two numbers are independent if and only if $p(x\cdot y)=p(x)\cdot p(y)$.} numbers.
\end{enumerate}
\end{definition}

Note that Copeland remarks that this second item is a consequence of the assumption that the sequence is obtained by independent trials (i.e. "the probability of success is a constant and does not vary from one trial to the next"). Church notes the connection between this notion of normal numbers and that of "completely normal number" by Armand Borel \cite{Borel}: 
\begin{quote}
These admissible numbers (to adopt Copeland's term) are closely
related to the normal numbers of Borel -- indeed an admissible number associated with the probability $\frac{1}{2}$ is the same as a number enti\`{e}rement normal to the base 2.
\end{quote}

\subsection{Postnikov and Pyateskii}

Around twenty years after Church's paper, a notion of of \emph{Bernoulli-normal} sequences was introduced by russian mathematicians, Postnikov and Pyateskii \cite{PostnikovPyateskii}. This notion coincides with the notion of $p$-distributed sequence defined above.

\begin{definition}[Bernoulli-normal sequence]\label{def:Bernoullinormal}
Let $\alpha\in\{0,1\}^{\omega}$ be a sequence $a_1,a_2,\dots,a_n,\dots$ and consider for every integer $s>0$ the $s$-th \emph{caterpillar} of x: 
\[\beta^s=(a_1,\dots,a_{s-1}),(a_2,\dots,a_{s}),\dots,(a_P,\dots,a_{P+s-1}),\dots.\]
The sequence $\alpha$ is \emph{Bernoulli-normal} if for any word ${w}$ of length $s$ with $j$ ones, $\freq[{w}]{\beta^s}$ exists and is equal to $p^j(1-p)^{s-j}$.
\end{definition}

In subsequent work, Postnikov \cite{Postnikov} considers the following alternative definition of admissible sequences. While this differs from Copeland's definition, we provide here a proof that the two notions coincide.

\begin{definition}[Postnikov admissible sequence]
Let $\alpha=a_1,a_2,\dots,a_n,\dots$ be a sequence in $\{0,1\}^{\omega}$. This sequence is called \emph{admissible} (in the sense of Postnikov) if for any word ${w}$ of length $m$ with $r_1,r_2,\dots,r_k$ the positions of its 1s, the sequence $\beta[{w}]= b_0,b_1,\dots,b_n,\dots$, defined by 
\[ b_n=a_{nm+r_1},a_{nm+r_2},\dots,a_{nm+r_k}, \]
satisfies that $\freq[1^k]{\beta[{w}]}$ exists and is equal to $p^k$.
\end{definition}

\begin{lemma}
A sequence $\alpha\in\{0,1\}^{\omega}$ is admissible in the sense of Copeland if and only if it is admissible in the sense of Postnikov.
\end{lemma}

\begin{proof}
Let $\alpha$ be a Postnikov-admissible sequence. Let us define the word ${u_i^n}=00\dots 010\dots 0$, of length $n$ with a single 1 at position $1\leqslant i\leqslant n$. Then $\freq[1]{\beta}$ exists and is equal to $p$. Since $\beta[{u_i^n}]=(i/n)\alpha$, this proves $\alpha$ satisfies the first item in Copeland's definition. The second item, namely the independance of the sequence $(1/n)\alpha, (2/n)\alpha, \dots, (n/n)\alpha$, is obtained by considering words ${u_{i,j}^n}$, of length $n$ with 1s exactly at the positions $i$ and $j$. Indeed, we have $\freq[11]{\beta[{u_{i,j}^n}]}=p^2=\freq[1]{\beta[{u_i^n}]}\freq[1]{\beta[{u_j^n}]}$, which coincide with Copeland's formalisation of independence.

Conversely, let $\alpha$ be a sequence, ${w}$ a word of length $m$ and $r_1,r_2,\dots,r_k$ the positions of the 1s in ${w}$. If $\alpha$ is Copeland-admissible, $\freq[1^k]{\beta[{w}]}=\prod_{i=1}^{k}\freq[1]{\beta[{u_{r_i}^n}]}$ by the requirement of independence, and therefore
\[ \freq[1^k]{\beta[{w}]}=\prod_{i=1}^{k}\freq[1]{(r_i/n)\alpha}=p^k\]
using that $\beta[{u_{r_i}^n}]=(r_i/n)\alpha$ and the first property of Copeland-admissible sequences, namely that $\freq[1]{(i/n)\alpha}=p$ for all $1\leqslant i\leqslant n$.
\end{proof}

Postnikov's then shows how the two notions, i.e. Bernoulli-normal and admissibility, coincide. However, the proof of Postnikov's theorem is -- to the authors' knowledge -- not available in english. As this result is related to the proof of Agafonov's theorem, we expect to include a translation in a later version of this document.

\begin{theorem}[Postnikov \cite{Postnikov}]\label{thm:Postnikov}
A sequence $\alpha\in\{0,1\}^{\omega}$ is Bernoulli-normal if and only if it is admissible.
\end{theorem}

\subsection{Postnikova}

A few years later, a short and beautiful paper by Postnikova characterises Bernoulli-normal sequences as the sequences for which the distribution of 1s is preserved by selecting strategies depending only on a finite number of preceding bits. In fact, Postnikova's result is the first to introduce finiteness and widely opens the way to Agafonov's theorem. It is stated as a new, restricted, notion of kollektiv. 

\begin{definition}[Postnikova-kollektiv]
Let $\alpha=a_1,a_2,\dots,a_n,\dots\in\{0,1\}^{\omega}$ be a sequence. The sequence $\alpha$ will be called a \emph{kollektiv} (in the sense of Postnikova) if:
\begin{enumerate}
\item $\freq{\alpha}$ exists and is equal to $p$;
\item for all word ${w}$ of length $s$, ${w}$ occurs in $\alpha$ an infinite number of times, and if a subsequence $\beta$ is made up consisting of the values immediately following the appearance of ${w}$ then $\freq{\beta}$ exists and is equal to $p$.
\end{enumerate}
\end{definition}

Note that using our own definition of kollektiv w.r.t. sets of strategies, a Postnikova-kollektiv is a kollektiv w.r.t. the set of strategies defined by a single finite word used as postfix, i.e. strategies $S_{{w}}$ defined as 
\[\{{v}\in \{0,1\}^* \mid \exists {u}, {v}={u}\cdot{w} \}.\]

\begin{theorem}[Postnikova]
A sequence $\alpha\in\{0,1\}^{\omega}$ is Bernoulli-normal if and only if it is a Postnikova-kollektiv.
\end{theorem}

The proof of this theorem can be found in the english translation \cite{Postnikova61} of Postnikova's paper \cite{Postnikova}. Note the error in translation in the definition of Postnikov-admissible sequences: the translator mentions \enquote{the relative frequency of appearances of ones in the sequence (2)}, while it should read \emph{the relative frequency of appearances of the word $1^k$ in the sequence (2)}. The confusion comes from the original russian formulation (which can be traced back to Postnikov's work \cite{Postnikov}) which is already ambiguous.

\subsection{Agafonov}

Agafonov's contribution was to relate this to the notion of automata. The main theorem of his original russian paper \cite{Agafonov} is stated as follows.

\begin{theorem}[Agafonov \cite{Agafonov}]\label{Agaftheorem}
A sequence $\alpha$ is normal if and only if it is a kollektiv w.r.t. the set of strategies computable by finite automata, i.e. it satisfies:
\begin{enumerate}
\item $\freq{\alpha}$ exists and is equal to $p$;
\item for all automata $M$, the subsequence $\beta$ consisting of the values immediately following the words accepted by $M$ is such that $\freq{\beta}$ exists and is equal to $p$.
\end{enumerate}
\end{theorem}

In fact, the proof of the implication from right to left in Agafonov's theorem is a consequence of Postnikova's theorem. Agafonov only refers to her work for this part of the proof. Indeed, if a sequence is a kollektiv in the sense of this theorem, it is also a Postnikova-kollektiv. Agafonov's contribution is therefore the proof of the converse implication, namely: if a sequence $\alpha$ is normal, it is a kollektiv w.r.t. the set of strategies computable by finite automata. 

However, the notion of normality used by Agafonov is not the notion of $p$-distributed sequence (\Cref{def:pdistrib}), or equivalently of Bernoulli-normal sequence (\Cref{def:Bernoullinormal}). Agafonov uses instead a notion of \emph{normality by blocks}.

\begin{definition}[Agafonov normal]\label{def:Agafonovnormal}
Let $\alpha\in\{0,1\}^{\omega}$ be a sequence $a_1,a_2,\dots,a_n,\dots$ and consider for every integer $s>0$ the $s$-th \emph{block sequence} of x: 
\[\beta^s=(a_1,\dots,a_{s-1}),(a_s,\dots,a_{2s-1}),\dots,(a_{ks},\dots,a_{(k+1)s-1}),\dots.\]
The sequence $\alpha$ is \emph{Agafonov-normal} if for any word ${w}$ of length $s$ with $j$ ones, $\freq[{w}]{\beta^s}$ exists and is equal to $p^j(1-p)^{s-j}$.
\end{definition}

This definition can be shown to be equivalent to Postnikov-admissibility which, combined with Postnikov's theorem (\Cref{thm:Postnikov}), proves the notion coincides with the usual notion of normality. 

\begin{lemma}
A sequence $\alpha$ is Agafonov-normal if and only if it is Postnikov-admissible.
\end{lemma}

\begin{proof}
In fact, the proof of this appears in the proof of Postnikov theorem, as Agafonov-normality is used as an intermediate notion. The proof of the right-to-left implication is taken from Postnikov's proof \cite{Postnikov}. The key observation is that the quantity $\freq[1^k]{\beta[{w}]}$ that appears in Postnikov-admissibility corresponds to the frequency of appearance of words in $\Delta$ in the sequence of blocks defined from $\alpha$, where $\Delta$ is the set of words ${u}$ of length $k$ that have 1s at these positions in which ${w}$ has 1s (but which may differ from ${w}$ on other bits). 

We first show that a Postnikov-admissible sequence $\alpha$ is Agafonov-normal. Let $\Sigma$ be the set of all length k word with fixed $\alpha$ bits equal to $1$ and $\beta$ bits equal to $0$, $\alpha+\beta\leq k$.
Write $T_l(\Sigma)$ the number of occurrences of $\Sigma$ in the sequence of blocks. Then by induction on $\beta$, using the definition of admissibility, we obtain:
\begin{equation} \lim_{l\rightarrow\infty} \frac{T_l(\Sigma)}{l}=p^\alpha q^{\beta}. \label{eq6Postnikov}\end{equation}
This gives the result by fixing $\Sigma$ as a singleton, i.e. $\alpha+\beta=k$.

Conversely, consider given an Agafonov-normal sequence. By definition, we know that the frequency of a word $w$ (with $j$ bits equal to $1$) is equal to $p^j q^{k-j}$. We want to sum this frequency over all words that have 1s at the same positions as $w$ but in which some 0s may have become 1s. I.e. we have all combinations of putting 1s in $k-j$ boxes. So the sum can be written as:
\[ \freq[1^k]{\beta[{w}]}=p^j \left( \sum_{k-j} \binom{i}{k-j} p^{i}q^{k-j-i}\right)= p^j (p+q)^{k-j} = p^j. \qedhere\]
\end{proof}

\subsection{The modern understanding of Agafonov's theorem}

It is important to note here that the original statement of Agafonov's Theorem \ref{Agaftheorem} differs widely from what is nowadays understood and referred as \emph{Agafonov's theorem}. Indeed, the name now refers to the following statement, which can be derived from Agafonov's theorem modulo \emph{compositionality of automata} (\autoref{prop:fin_sel_comp}).

\begin{theorem}[Agafonov's theorem, modern understanding]
Let $\alpha$ be a normal sequence. Any infinite subsequence selected by a finite automata is again normal.
\end{theorem}

\begin{proof}
Let $\alpha$ be a normal sequence, and $M$ an automata selecting an infinite subsequence $\beta$. By Agafonov's theorem \ref{Agaftheorem}, the frequency of $1$s in $\beta$ is defined and equal to $p$. By \autoref{prop:fin_sel_comp}, for any automata $N$ there exists an automata $N\circ M$ such that $N\circ M[\alpha]=N[\beta]$. So for any automata $N$ such that the selected subsequence $N[\beta]$ is infinite, this subsequence is also a subsequence selected from $\alpha$, hence Agafonov's Theorem \ref{Agaftheorem} implies that the frequency of $1$s in $N[\beta]$ is defined and equal to $1$.

This just proves that $\beta$ is an Agafonov kollectiv. A last application of Agafonov's theorem \ref{Agaftheorem} then implies that $\beta$ is normal, proving the theorem.
\end{proof}


\section{Agafonov's original proof: a direct translation}

We fix once and for all the alphabet $\Sigma=\{0,1\}$.

\begin{definition}
	Let ${a}=a_{1}a_{2}\dots a_{n}$ be a word over $\Sigma$. We define: 
	\[ \mu_{p}({a})=p^{\countones{a}}(1-p)^{n-\countones{a}} \]
\end{definition}

\begin{definition}
For $M\subseteq\{0,1\}^{N}$, de define $\mu_{p}(M)=\sum_{{w}\in M}\mu_{p}({w})$.
\end{definition}

\begin{definition}
Let $\alpha=a_1,a_2,\dots, a_n,\dots$ be a sequence in $\{0,1\}^\omega$. For all natural number $n$ we define the $n$-block decomposition of $\alpha$ as the sequence $(\alpha_{(n,r)})_{r\geqslant 1}$ defined by
\[ \alpha_{(n,r)}=a_{n(r-1)+1}a_{n(r-1)+2}\dots x_{nr} \]
\end{definition}

\begin{definition}
	Let $\alpha$ be a sequence in $\{0,1\}^\omega$, ${w}$ a finite word of length $n$, and $k$ an integer. We define $\fqcy[\alpha]{{w}}{k}=\frac{1}{k}\card\{\alpha_{(n,r)}={w}~|~r\leqslant k\}$.
\end{definition}

Notice that a sequence $\alpha$ is Agafonov-normal (\Cref{def:Agafonovnormal}) if and only if for all finite word ${w}$ of length $n$ with $j$ bits equal to $1$, $\lim_{k\rightarrow\infty} \fqcy[\alpha]{{w}}{k}$ exists and is equal to $p^j q^{n-j}$.

\begin{definition}
Let $A$ be a strongly connected automata with set of states $Q$. For all $q\in Q$, we write $A_{q}$ the automata $A$ in which the state $q$ is chosen as initial.
\end{definition}

\begin{definition}
Let $A$ be a strongly connected automata with set of states $Q$, and $q\in Q$. Let ${w}=w_1w_2\dots w_n$ be a finite word. We write $\pickedout{{w}}{A_{q}}$ the word \emph{picked out} by the automata $A_{q}$, i.e. the word $w_{i_1}w_{i_2}\dots w_{i_k}$ where $i_1<i_2<\dots<i_k$ is the increasing sequence of indices $1\leqslant j\leqslant n$ such that ${w}\vert_{\leq j-1}$ is accepted by $A_{q}$.
\end{definition}

\begin{definition}
Let $A$ be a strongly connected automata with set of states $Q$. For all $p\in [0,1]$, $b\in [0,1]$, $n\in\naturalN$ and $\epsilon>0$, we define the sets:
\[ D_{n}^{p}(b,\epsilon)=\{{w}\in\{0,1\}^{n}~|~ \forall q\in Q, \len{\pickedout{{w}}{A_{q}}}>bn, \abs{\frac{\countones{\pickedout{x}{A_{q}}}}{\len{\pickedout{x}{A_{q}}}}-p}<\epsilon\} \]
\end{definition}

\begin{claim}\label{mainclaim}
For all $\epsilon>0$ and all $p\in[0,1]$, $\lim_{n\rightarrow\infty} \mu_{p}(D_{n}^{p}(b,\epsilon)) = 1$.
\end{claim}

\begin{proof}
This claim is a consequence of \Cref{lemma1} and \Cref{lemma2} below, noting that $D_{n}^{p}(b,\epsilon)=\Sigma^n\diagdown(E_n(b)\cup G_n(b,\epsilon))$.
\end{proof}


\begin{theorem}
Let $\alpha$ be a normal sequence with ratio $p\in[0,1]$, $A$ a strongly connected automata. Then the sequence $\beta=\pickedout{\alpha}{A}$ is normal with ratio $p$.
\end{theorem}

\begin{proof}
We will show that $\forall\epsilon, \exists L, \forall l\geqslant L, \abs{\frac{1}{l}\sum_{i=1}^{l}\seq{y}_{i} -p}<\epsilon$. 

Pick $\delta>0$ small enough ($\delta<\frac{b\epsilon}{8}$). By \Cref{mainclaim}, we pick $n\in\naturalN$ such that $\mu_{p}(D_{n}^{p}(b,\epsilon))>1-\delta$. Now, we consider $\eta<\frac{b\epsilon}{8}$ (i.e. sufficiently small); since $\alpha$ is normal, there exists $S\in\naturalN$ such that $\forall s\geqslant S$, $\forall {a}\in\{0,1\}^{n}$, $\abs{\fqcy[\alpha]{a}{s}-\mu_{p}({a})}<\frac{\eta}{2^{n}}$, i.e. $\forall M\subseteq\{0,1\}^{n}$, $\abs{\fqcy[\alpha]{M}{s}-\mu_{p}(M)}<\eta$.

We now consider the sequence $\beta_{[n,r]}$ as the sequence of blocks of $\pickedout{\alpha}{A}$ (of changing length between $0$ and $n$) corresponding to the sequence of blocks $\alpha_{(n,r)}$, and write $\theta$ the frequency of 1s in the blocks picked out from the blocks in $D_{n}^{p}(b,\frac{\epsilon}{2})$. Then $\abs{\theta-p}<\frac{\epsilon}{2}$.

Now let $L=\sum_{i=1}^{s}\len{\beta_{[n,i]}}$ and $\ell=\sum_{i\in I}\len{\beta_{[n,i]}}$ with $I=\{i\leqslant s~|~ \alpha_{(n,i)}\not\in D_{n}^{p}(b,\frac{\epsilon}{2})\}$. We write $\theta=\frac{\sum_{i\in I} \countones{\beta_{[n,i]}}}{\sum_{i\in I}\len{\beta_{[n,i]}}}$ and $\rho=\frac{\sum_{i=1}^{s} \countones{\beta_{[n,i]}}}{L}$. Then $\abs{\rho-\theta}<\frac{\ell}{L}$.

We then show $\frac{\ell}{L}<\frac{\epsilon}{2}$ and deduce that $\abs{\rho-p}<\epsilon$ as follows.
We consider a small enough $\delta>0$ and find $S$ big enough to have
\[ \frac{\card\{i\leqslant S~|~ \alpha_{[n,i]}\in D_{n}^{p}(b,\frac{\epsilon}{2}))\}}{S}>1-\delta-\eta \]
On one hand, for all ${w}\in D_{n}^{p}(b,\frac{\epsilon}{2}))$ more than $bn$ characters are picked out, therefore we have $L>(1-\delta-\eta)Sbn$. On the other hand, for all ${w}\in \{0,1\}^{n}$ less than $n$ characters are picked out and $\frac{\card\{i\leqslant S~|~ \alpha_{[n,i]}\not\in D_{n}^{p}(b,\frac{\epsilon}{2}))\}}{S}<\delta+\eta$, thus $\ell<(\delta+\eta)Sn$. Hence $\frac{\ell}{L}<\frac{(\delta+\eta)}{(1-\eta-\delta)b}<\frac{\epsilon}{2}$.

Finally, $\abs{\rho-p}\leqslant\abs{\rho-\theta}+\abs{\theta-p}<\epsilon$.
\end{proof}

\begin{lemma}\label{lemma1}
Define $E_{n}(b,q)=\{{w}\in\{0,1\}^{n}~|~ A_{q}[{w}]\leqslant bn\}$, and $E_{n}(b)=\cup_{q\in Q} E_{n}(b,q)$. Then for all $p\in [0,1]$ and for all automaton $A$, there exists $c,d>0$ such that for all $\epsilon>0$, the following holds.
\[ \lim_{n\rightarrow\infty} \mu_{p}(E_{n}(\frac{c-\epsilon}{d}))=0 \] 
\end{lemma}

\begin{proof}
Let us consider $(X,\mathcal{B},\mu_{p})$ the measure space with $X=\{0,1\}^{\omega}$, $\mathcal{B}$ induced by cylinders, and $\mu_{p}(\{\alpha~|~\forall j \in\{1,2,\dots,n\}, \alpha_{i_{j}}=b_{j}\})=\mu_{p}(b_{1}b_{2}\dots b_{n})$.

For a word ${v}$, define $C({v})=\{\alpha\in\{0,1\}^{\omega}~|~\exists \beta\in\{0,1\}^{\omega}, \alpha={v}.\beta\}$.
If $R$ is a finite (prefix-free\footnote{This precision is added by the authors.}) set of words, then 
\begin{equation}\label{techlemma}
\mu_{p}(\cup_{{v}\in R} C({v}))=\mu_{p}(R).
\end{equation}

Now, take $A$ a finite automaton $(\{0,1\},Q,Q^{\ast},\phi)$. This defines a Markov chain of set of states $Q$:
\[ p_{i,j}=\left\{\begin{array}{ll}
	1 & \text{ if }\phi(i,1)=\phi(i,0)=j\\
	p & \text{ if }\phi(i,1)=j, \phi(i,0)\neq j\\
	1-p & \text{ if }\phi(i,1)\neq j, \phi(i,0)=j\\
	0 & \text{ otherwise}\\
	\end{array}\right.
	\]
If $A$ is strongly connected, there exists a smallest $n_{i,j}$ such that $p_{i,j}^{(n_{i,j})}>0$. Define the period $D$ as the least common multiple of the family $(n_{i,j})_{i,j\in Q^{2}}$.


Let $Q_{0},Q_{1},\dots,Q_{D-1}$ be the classes of \enquote{periodical states}. Given $Q_{r}$, we have a Markov chain with probabilities $p_{i,j}^{(D)}$ for $i,j\in Q_{r}$. For all $Q_{r}$, there exists a family $(c_{i})_{i\in Q_{r}}$ such that $\sum_{i\in Q_{r}} c_{i}=1$ and $\lim_{n\rightarrow \infty} p_{i,j}^{(Dn)}=c_{j}$.

Consider $q_{A_{j}}(\alpha)=q_{1}q_{2}\dots$ the realisation of the Markov process with $\alpha$ as input and $j\in Q$. We have 
\begin{equation}\label{eq1}
\mu_{p, A_{j}}(\{q_{A_{j}}(\alpha)~|~\alpha\in M\})=\mu_{p}(M).
\end{equation}

Let $\nu_{i}^{(n)}(\vec{q})=\card\{q_{j}=i~|~j\leqslant n\}$. For all $\epsilon>0$ and all $i, j$,
\begin{equation}\label{eq2}
\lim_{n\rightarrow \infty} \mu_{p,A_{j}}\{\vec{q}\text{ s.t. }\abs{\frac{d}{n}\nu_{i}^{(n)}(\vec{q})-c_{i}}\geqslant \epsilon\}=0
\end{equation}
by the \emph{law of large numbers for finite regular ergodic Markov chains}.

From \Cref{eq1} and \Cref{eq2}, we have
\[ \lim_{n\rightarrow\infty} \mu_{p}\{\alpha\text{ s.t. }\abs{\frac{D}{n}\nu_{i}^{(n)}(q_{A_{j}}(\vec{x}))-c_{i}}\geqslant\epsilon\}=0 \]
For a finite word ${a}$, write $q_{A_{j}}({a})=q_{1}q_{2}\dots q_{n}$ ($n={\rm len}({a})$). Using \Cref{techlemma}, 
\begin{equation}\label{eq3}
\lim_{n\rightarrow\infty} \mu_{p}\{ a_{1}a_{2}\dots a_{n}, \abs{\frac{D}{n}\nu_{i}^{(n)}(q_{A_{j}}(a_{1}a_{2}\dots a_{n}))-c_{i}}\geqslant\epsilon\}=0
\end{equation}
If in $q_{A_{j}}({a})$, there exists $q_{i}\in Q^{\ast}$, then $A_{j}$ picks out $a_{j}$ from ${a}$. Let $c=\min_{i\in Q^{\ast}} c_{i}$. From \Cref{eq3}, for all $j\in Q$, $\lim_{n\rightarrow\infty}\mu_{p}E_{n}(\frac{c-\epsilon}{D},j)=0$.

The lemma then follows from $\mu_{p}E_{n}(\frac{c-\epsilon}{D})\leqslant \sum_{j\in Q}\mu_{p}E_{n}(\frac{c-\epsilon}{D},j)$.
\end{proof}

\begin{lemma}\label{lemma2}
Define $G_{n}(b,\epsilon,q)=\{{w}\in\{0,1\}^{n}~|~\len{A_{q}[{w}]}>bn, \abs{\frac{\countones{A_{q}[{w}]}}{\len{A_{q}[{w}]}}-p}>\epsilon\}$, and $G_{n}(b,\epsilon)=\cap_{q\in Q}G_{n}(b,\epsilon,q)$. Then for all $p,b,\epsilon$ and all automaton $A$, 
\[\lim_{n\rightarrow\infty}\mu_{p}(G_{n}(b,\epsilon))=0.\]
\end{lemma}

\begin{proof}
(Similar to Lemma 3 from D.W. Loveland, \emph{The Kleene hierarchy classification of recursively random sequences}  \cite{Loveland66}.)

By the "strong law of large numbers", for all $\epsilon>0$, 
\[ \lim_{n\rightarrow\infty} \mu_{p}(\cup_{\ell\geqslant n}\{\seq{y}~|~\abs{\frac{1}{\ell}\sum_{i=1}^{\ell}y_{i}-p}\geqslant\epsilon\})=0 \]
Define $F_{n}(b,\epsilon)=\cup_{bn<\ell\leqslant n}\{{y}\in\{0,1\}^{\ell}~|~\abs{\frac{1}{\ell}\sum_{i=1}^{\ell}y_{i}-p}\geqslant\epsilon\}$. And define $R_{n}(b,\epsilon)$ as the set obtained from $F_{n}(b,\epsilon)$ by removing the words ${w}$ such that there exists a word ${u}$ in $F_{n}(b,\epsilon)$ with ${u}\prec{w}$.

From the fact that 
\[ \cup_{{w}\in R_{n}(b,\epsilon)}\mathcal{C}({w}) \subset \cup_{\ell\geqslant bn}\{\seq{y}~|~\abs{\frac{1}{\ell}\sum_{i=1}^{\ell}y_{i}-p}\geqslant\epsilon\} \]
and 
\[ \mu_{p}(R_{n}(b,\epsilon))=\mu_{p}(\cup_{{w}\in R_{n}(b,\epsilon)}\mathcal{C}({w})) \]
and the equation above, we deduce that 
\[ \lim_{n\rightarrow\infty} \mu_{p}R_{n}(b,\epsilon)=0. \]
By \Cref{lemma3} and the equality
\[ G_{n}(b,\epsilon,q) = \{{w}\in\{0,1\}^{n}~|~\pickedout{w}{A_{q}}\in S_{n}(b,\epsilon)\}\]
we get that $\mu_{p}(G_{n}(b,\epsilon,q))\leqslant \mu_{p}(R_{n}(b,\epsilon))$.
Consequently, $\lim_{n\rightarrow\infty} \mu_{p}(G_{n}(b,\epsilon,q))=0$ for all $q\in Q$, hence $\lim_{n\rightarrow\infty} \mu_{p}(G_{n}(b,\epsilon))=0$.
\end{proof}

\begin{lemma}\label{lemma3}
Let $S$ be a strategy, and $F$ a finite subset of $\{0,1\}^{\ast}$. Let $R$ be the set obtained from $F$ by removing those words ${w}$ such that there exists a word ${u}\in F$ with ${u}\prec {w}$ (i.e. ${u}$ is a prefix of ${w}$). Let $M$ be the set $\{{w}\in\{0,1\}^{n}~|~ S({w})\in F\}$. Then
\[\mu_{p}(M)\leqslant \mu_{p}(R).\]
\end{lemma}

\begin{proof}
It is sufficient to prove that for a given word ${w}=a_{1}a_{2}\dots a_{k}$ the set $M=\{{u}\in \{0,1\}^{n}~|~{w}\preceq S({u})\}$ satisfies $\mu_{p}(M)\leqslant \mu_{p}({w})$. This is show by induction on the length $k$ of the word ${w}$.\newline
The base case is ${w}=a_{1}=1$ (by symmetry -- 1 becomes 0, $p$ becomes $1-p$ --, this is sufficient). Let $\alpha=x_{1}x_{2}\dots x_{n}$ be a word in $M$, and write $x_{f}$ the first symbol picked out by $S$; in particular $x_{f}=1$. Now, one can define $\bar{\alpha}=x_{1}x_{2}\dots x_{f-1}\bar{x}_{f}x_{f+1}\dots x_{n}$, i.e. the word obtained from $\alpha$ by simply flipping the $f$-th bit. Then $\bar{\alpha}\not\in M$. One can then define the set $\bar{M}=\{\bar{\alpha}~|~ \alpha\in M\}$. As $\mu_{p}(\bar{\alpha})=\frac{1-p}{p}\mu_{p}(\alpha)$ and $\bar{\cdot}$ defines a one-to-one correspondence between $M$ and $\bar{M}$, we have $\mu_{p}(\bar{M})=\frac{1-p}{p}\mu_{p}(M)$. Moreover, $M$ and $\bar{M}$ are disjoint subsets of $\{0,1\}^{n}$, hence $\mu_{p}(M)\leqslant 1-\mu_{p}(\bar{M})$. We can then conclude from these two equations that $\mu_{p}(M)\leqslant p$.

Now, consider the word ${w}=a_{1}a_{2}\dots a_{k}a_{k+1}$ with $a_{k+1}=1$. We have $M=\{\alpha\in\{0,1\}^{n}~|~ a_{1}\dots a_{k}1\preceq S(\alpha)\}$. Given $\alpha\in M$, define $\bar{\alpha}$  as the word obtained from $\alpha$ by flipping its $k+1$-th picked out bit, i.e. $\bar{\alpha}$ is the unique word obtained from $\bar{x}$ by flipping a single bit and such that $a_{1}\dots a_{k}0\preceq S(\bar{\alpha})$. Define $\bar{M}$ as the set $\{\bar{\alpha}~|~\alpha\in M\}$. Let $N$ be the set $\{\alpha\in\{0,1\}^{n}~|~ a_{1}\dots a_{k}\preceq S(\alpha)\}$. Then $N$ contains both $M$ and $\bar{M}$, and the latter two sets are disjoint. Moreover the induction hypothesis implies that $\mu_{p}(N)\leqslant \mu_{p}(a_{1}a_{2}\dots a_{k})$. Hence $\mu_{p}(M)+\mu_{p}(\bar{M})\leqslant  \mu_{p}(a_{1}a_{2}\dots a_{k})$. Since $\mu_{\bar{M}}=\frac{1-p}{p}\mu_{p}(M)$, we deduce that $\mu_{p}(M)\leqslant p\mu_{p}(a_{1}a_{2}\dots a_{k})= \mu_{p}(a_{1}a_{2}\dots a_{k}1)$.
\end{proof}

\section{Adaptation of Agafonov's proof}

We now give an embellished, modern account of Agafonov's proof; we have endeavoured to use pedagogical explanations and have extended the treatment to make the text more readily readable the the modern reader.

 \begin{definition}
 A \emph{finite-state selector} over $\{0,1\}$ is a deterministic finite automaton
 $S = (Q,\delta,q_s,Q_F)$ over $\{0,1\}$. A finite-state selector is strongly connected
 if its underlying directed graph (states are nodes, transitions are edges) is strongly connected.
 Denote by $L(S)$ the language accepted by the automaton.
 
 If $\alpha = a_1 a_2 \cdots $ is a finite or right-infinite sequence over $\{0,1\}$, the 
 subsequence \emph{selected by} $A$ is the (possibly empty) sequence of letters
 $a_n$ such that the prefix $a_1 \cdots a_{n-1} \in L(S)$, that is,
 the automaton when started on the finite word $a_1 \cdots a_{n-1}$ in state $q_s$
ends in an accepting state after having read the entire word.
 \end{definition}

For two
words ${u},{v}$, we write
${u} \preceq {v}$ if ${u}$ is a prefix
of $\prec{v}$, and ${u} \prec {v}$ if 
${u}$ is a proper prefix of ${v}$.


\begin{definition}
Let ${a} = a_1 \cdots a_n$ and ${b} = b_1 \cdots b_N$ be finite words over $\{0,1\}$. We denote by
$\countones{{a}}{{b}}$ the number of occurrences of ${a}$ in ${b}$, that is, the quantity
$$
\left\vert \left\{j : b_j b_{j+1} \cdots b_{j+n-1} = a_1 a_2 \cdots a_n \right\}\right\vert
$$
\end{definition}


\begin{definition}
	Let ${a}=a_{1}a_{2}\dots a_{n}$ be a word over $\{0,1\}$, and $p$ a probability distribution on $\{0,1\}$. We define: 
	\[ \mu_{p}({a})=\prod_{i=1}^n p(a_i) \]

If $M\subseteq\{0,1\}^*$ is finite, we define $\mu_{p}(M)=\sum_{{w}\in M}\mu_{p}({w})$ (and set
$\mu_p(\emptyset) = 0$).
\end{definition}


\begin{definition}
Let $\alpha=x_{1}x_{2}\dots x_{n} \cdots$ be a sequence over $\{0,1\}$. We say that $\alpha$ is $p$-\emph{block-distributed} if, for each $n \geqslant 1$ and every ${w} \in \{0,1\}^n$,
the $n$-block decomposition $(\alpha_{(n,r)})_{r\geqslant 1}$ of $\alpha$ satisfies:
$$
\lim_{k \rightarrow \infty} \frac{\vert i \leq k : \alpha_{(n,k)} = {w} \vert}{k} = \mu_p({w})
$$
\end{definition}

As already remarked above, this notion coincides with Agafonov-normality (\Cref{def:Agafonovnormal}).


\begin{remark}\label{rem:positive_definite}
Like in Agafonov's original paper, for a finite-state selector $A$, we \emph{do not} require that all cycles in the underlying directed graph of $A$ contain at least one accepting state. This assumption is occasionally made in modern papers on Agafonov's Theorem
to ensure that if ${w} \in \{0,1\}^\omega$ is a normal sequence, then $A[{w}]$ is infinite as well. But just as in Agafonov's paper, the requirement turns out to be unnecessary (see \Cref{main:lemma1}).

However, in Agafonov's paper, the probability $\mu_p(1)$ of obtaining a $1$ was assumed to satisfy
$0 < \mu_p(1) < 1$ (i.e., both $0$ and $1$ occur with positive probability).
Without this assumption, there are connected automata that fail to pick out infinite sequences from $p$-distributed ones. For example, define 
$A = (\{q_0,q_1\},\{0,1\},\delta,q_0,\{q_0\})$
where 
$$
\begin{array}{lr}
\delta(q_0,0) = q_0 & \delta(q_0,1) = q_1\\
\delta(q_1,1) = q_0 & \delta(q_0,1) = q_0
\end{array}
$$
Define $\mu_p(0)=1$ and $\mu_p(1) = 0$. Then,
${w} = 10^\omega$ is $p$-distributed, but
$A[{w}] = 0$, hence is finite.
\end{remark} 

Motivated by \Cref{rem:positive_definite},
we have the following definition:

\begin{definition}
A Bernoulli distribution $p : \{0,1\} \longrightarrow [0,1]$ is said to be \emph{positive} if,
for all $a \in \{0,1\}$, $p(a) > 0$. The probability map $\mu_p : \{0,1\}^* : \longrightarrow [0,1]$ is positive if $p$
is positive.
\end{definition}



\begin{proposition}[Finite-State selectors are compositional]
\label{prop:fin_sel_comp}
Let $A$ and $B$ be DFAs over the same alphabet. Then there is a DFA $C$ such that, for each
sequence ${w}$, $\pickedout{{w}}{C} = \pickedout{\pickedout{{w}}{A}}{B}$.
\end{proposition}

\begin{proof}
Let $A = (Q^A,\{0,1\},\delta^A,q_0^A,F^A)$
and $B = (Q^B,\{0,1\},\delta^B,q_0^B,F^B)$.
Define $Q^C = Q^A \times Q^B$,
and set $q_0^C = (q_0^A,q_0^B)$
and $F^C = F^A \times F^B$.
For each $q^B \in Q^B$,
define the set $D_{q^B} = \{(q,q^B) : q \in Q^A \}
\subseteq Q^C$. Observe that
$Q^C = \bigcup_{q^B \in Q^B} D_{q^B}$ and
that for $q^B, r^B \in Q^B$ with
$q^B \neq r^B$, we have $D_{q^B} \cap D_{r^B} = \emptyset$, and thus $\{D_{q^B} : q^B \in Q^B\}$ is a partitioning of $Q^C$.
Hence, the transition relation, $\delta^C$, of $C$ may be defined by 
defining it separately on each subset $D_{q^B}$:
$$
\delta^C((q,q^B),a) = \left\{
\begin{array}{ll}
    (r,q^B) & \textrm{if } q \notin F^A  \textrm{ and } \delta^A(q,a) = r \\
     (r,r^B) & \textrm{if } q \in F^A \textrm{ and } \delta^A(q,a) = r
        \textrm{ and } \delta^B(q^B,a) = r^B \\

\end{array}
\right.
$$
Thus, when $C$ processes its input, it freezes the current state $q^B$
of $B$ (the freezing is represented by staying within $D_{q^B}$) and simulates
$A$ until an accepting state of $A$ is reached (i.e. just before $A$ would select the
next symbol); on the next transition, $C$ unfreezes the current state of $B$
and moves to the next state $r^B$ of $B$ and then freezes it and continues with a
simulation of $A$.

Observe that a symbol is picked out by $C$ if{f} the state is an element of $F^C = F^A \times F^B$ if{f} the symbol is the next symbol
read after simulation of $A$ reaches an accepting state of $A$ when
the current frozen state of $B$ is an accepting state of $B$.
\end{proof} 

%

The following shows that to prove that $p$-distributedness is preserved
under finite-state selection, it suffices to prove that the
limiting frequency of each $a \in \{0,1\}$ exists
and is equal to $p(a)$.

\begin{lemma}\label{lem:simply_normal_is_enough}
Let $\alpha$ be a $p$-distributed sequence. The following are equivalent:

\begin{itemize}

\item For all connected DFAs $A$, $A[\alpha]$ is $p$-distributed.

\item For all connected DFAs $A$ and all $a \in \{0,1\}$, 
the limiting frequency of $a$ in $A[\alpha]$ exists and 
is equal to $p(a)$.

\end{itemize}
\end{lemma}

\begin{proof}
If, for all $A$, $A[\alpha]$ is $p$-distributed, then in particular
the limiting frequency of $a$ in $A[\alpha]$ exists and 
is equal to $p(a)$ for all $A$.

Conversely, suppose that, for all DFAs $A$ and all $a \in \{0,1\}$, 
the limiting frequency of $a$ in $A[\alpha]$ exists and 
is equal to $p(a)$.
We will prove by induction on $k \geq 0$
that the limiting frequency of 
every $v_1 \cdots v_k v_{k+1} \in \{0,1\}^{k+1}$ exists and equals
$p(v_1 \cdots v_k v_{k+1})$.

\begin{itemize}
    \item $k = 0$: This is the supposition.
    
    \item $k \geq 1$. Suppose
    that the result has been proved
    for $k-1$. Let $v_1 \cdots v_k \in \{0,1\}^k$; by the induction hypothesis, the limiting frequency of $v_1 \cdots v_k$ in $A[{w}]$
    is $p(v_1 \cdots v_k)$. We claim that there is a strongly connected DFA
    $B$ that, from any sequence, selects the symbol after each occurrence
    of $v_1 \cdots v_k$. To see that such a DFA exists,  let there be a state for each element of $\{0,1\}^k$
and assume that the state is the current length-$k$ string in a ``sliding window'' that moves over ${w}$ one symbol at the time; when the window is moved one step, the DFA transits to the state representing the new
length-$k$ string in the window, i.e. from the state representing the word $w_1 \cdots w_k$, there are transitions to
$w_2 \cdots w_k 0$ and $w_2 \cdots w_k 1$; it is easy to see that each state is reachable from every other state in at most $k$ transitions. The unique final state of $B$ is the state representing $v_1 \vdots v_k$; the start state of $B$
can be chosen to be any state representing a string $w_1 \cdots w_k$ such that there is exactly $k$ transitions to the final state.

    By \Cref{prop:fin_sel_comp},
    there is a connected DFA $C$
    such that $C[{w}] = B[A[{w}]]$.
    
    For any
    $a \in \{0,1\}$ and
    any sufficiently large positive integer $N$, we have
    $$
    \frac{\countones{a}{\pickedout{{w}_{\leq N}}{C}}}{\vert \pickedout{{w}_{\leq N}}{C} \vert} =
    \frac{\countones{a}{\pickedout{\pickedout{{w}_{\leq N}}{A}}{B}}}{\vert \pickedout{\pickedout{{w}_{\leq N}}{A}}{B} \vert} =  = 
    \frac{\countones{v_1 \cdots v_k a}{\pickedout{{w}_{\leq N}}{A}}}{\countones{v_1 \cdots v_k }{\pickedout{{w}_{\leq N}}{A}}}
    $$

As $C$ is connected, there is a 
real number $b$ with $0 < b \leq 1$ such that $C$ selects
at least $bN$ symbols from
${w}_{\leq N}$, and 
by the induction hypothesis,
for every $\epsilon > 0$,
there is an $M$ such that for all
$N > M/b$,
$\left\vert \frac{\countones{a}{\pickedout{{w}_{\leq N}}{C}}}{\vert \pickedout{{w}_{\leq N}}{C} \vert}
- p(a) \right\vert < \epsilon$ and hence
$\left\vert\frac{\countones{v_1 \cdots v_k a}{\pickedout{{w}_{\leq N}}{A}}}{\countones{v_1 \cdots v_k }{\pickedout{{w}_{\leq N}}{A}}}
- p(a) \right\vert < \epsilon$.

But for all sufficiently large $N$,
the induction hypothesis furnishes
$$
\left\vert \frac{\countones{v_1 \cdots v_k }{\pickedout{{w}_{\leq N}}{A}}}{\vert \pickedout{{w}_{\leq N}}{A}\vert}
 - p(v_1 \cdots v_k)\right\vert < \epsilon
 $$
But as
 $$
 \frac{\countones{v_1 \cdots v_k a}{\pickedout{{w}_{\leq N}}{A}}}{\vert \pickedout{{w}_{\leq N}}{A}\vert}
 =
 \frac{\countones{v_1 \cdots v_k a}{\pickedout{{w}_{\leq N}}{A}}}{\countones{v_1 \cdots v_k }{\pickedout{{w}_{\leq N}}{A}}}
 \cdot \frac{\countones{v_1 \cdots v_k }{\pickedout{{w}_{\leq N}}{A}}}{\vert \pickedout{{w}_{\leq N}}{A}\vert}
 $$
 we hence have
 (as $p(v_1 \cdots v_k)p(a) = p(v_1 \cdots v_k a)$ because $p$ is Bernoulli):
 \begin{eqnarray*}
 \lefteqn{\left\vert \frac{\countones{v_1 \cdots v_k a}{\pickedout{{w}_{\leq N}}{A}}}{\vert \pickedout{{w}_{\leq N}}{A}\vert} - p(v_1 \cdots v_k a)\right\vert} \\
 &<&
 \epsilon^2 + \epsilon\left( \frac{\countones{v_1 \cdots v_k a}{\pickedout{{w}_{\leq N}}{A}}}{\countones{v_1 \cdots v_k }{\pickedout{{w}_{\leq N}}{A}}}
 + \frac{\countones{v_1 \cdots v_k }{\pickedout{{w}_{\leq N}}{A}}}{\vert \pickedout{{w}_{\leq N}}{A}\vert}\right)\\
 &\leq& \epsilon^2 + 2\epsilon
 \end{eqnarray*}
 Hence, for all $a \in \{0,1\}$, the limiting frequency of $v_1 \cdots v_k a$ in $\pickedout{{w}_{\leq N}}{A}$  exists and equals $pv_1 \cdots v_k a$,
 as desired.
\end{itemize}
\end{proof}

\begin{definition}
A strategy $S$ is a predicate over the set of finite words, i.e. $S\subseteq \{0,1\}^{*}$.

Given a strategy $S$ and a right-infinite sequence $\seq{x}$ in $\{0,1\}^{\omega}$, we define the sequence $S(\seq{x})$ as follows. Let $i_{1},i_{2},\dots,i_{k},\dots$ be the (increasing) sequence of indices $j$ such that $\seq{x}_{<j}\in S$ and $S(\seq{x})_{j}= \seq{x}_{i_{j}}$.
\end{definition}

Thus, $S({w})$ is simply the subsequence
of ${w}$ that are ``picked out'' by applying $S$ 
to prefixes of ${w}$. Note also
that if ${w} \in S$, then 
in any word on the form
${w} \cdot b \cdot {v}$, then
$S$ must pick $b$. Thus, $S$ cannot be made
to, for instance, only pick out $0$ or $1$--it picks out ``the next symbol'' after any
${w} \in S$.

\begin{definition}
Let $A = (Q,\{0,1\},\delta,q_0,F)$ be a connected DFA. For all $q\in Q$, we denote by $A_{q}$ the automaton
$(Q,\{0,1\},\delta,q,F)$,
i.e. where the state $q$ is chosen as the initial state.
\end{definition}

\begin{definition}
Let $A = (Q,\{0,1\},\delta,q_0,F)$ be a connected DFA, and let $q\in Q$. Let $\alpha$ be a right-infinite sequence over $\{0,1\}$. We denote by $\pickedout{x}{A_{q}}$ the subsequence $\bar{\alpha}$ of $\alpha$ \emph{picked out} by $A_{q}$, that is, $w_{i}\in\bar{{w}}$ if and only if $A_{q}({w}_{<i})$ reaches an accepting state.
\end{definition}

For every fixed positive integer $n$, it is clear that
$(\{0,1\}^n,\textrm{Pr})$ is a finite
probability space
where $\textrm{Pr}(M) = \mu_p(M)$
for every $M \subseteq \{0,1\}^n$.

\begin{definition}
Let $A = (Q,\{0,1\},\delta,q_0,F)$ be a strongly connected DFA. For all $p\in [0,1]$, $b\in [0,1]$, $n\in\naturalN$ and $\epsilon>0$, we define sets $D^p_n(b,\epsilon)$,
$E_n(b,q)$ and $G_n(b,\epsilon,q)$ as follows:
{\small
\begin{align}
D_{n}^{p}(b,\epsilon,q) &= \left\{{w}\in\{0,1\}^{n}~:~  \vert\pickedout{{w}}{A_{q}} \vert >bn \textrm{ and }  \left\vert \frac{\countones{\pickedout{w}{A_{q}}}}{\vert\pickedout{w}{A_{q}}\vert}-p(a)\right\vert<\epsilon\right\} \\
D_n^p(b,\epsilon) &= \bigcap_{q \in Q} D_n^p(b,\epsilon,q)\\
E_n(b,q) &= \{{w}\in\{0,1\}^{n} : \vert A_{q}[{w}] \vert\leq bn\} \\
E_{n}(b) &= \bigcup_{q\in Q} E_{n}(b,q)\\
G_n(b,\epsilon,q) &= \left\{{w}\in\{0,1\}^{n} : \vert A_{q}[{w}]\vert >bn \textrm{ and }  \left\vert\frac{\countones{A_{q}[{w}]}}{\vert A_{q}[{w}] \vert}-p(a) \right\vert \geq\epsilon\right\}\\
G_{n}(b,\epsilon) &=\bigcup_{q\in Q}G_{n}(b,\epsilon,q)
\end{align}
}
\end{definition}

Observe that, for all $b,n,\epsilon$,
$$
\{0,1\}^n = E_n(b) \cup D_n^p(b,\epsilon) \cup G_n(b,\epsilon)
$$
\noindent (but $E_n(b)$ and $G_n(b,\epsilon)$ are not necessarily disjoint).

\begin{lemma}\label{main:lemma1}
Let $A=(Q,\{0,1\},\delta,q_0,F)$ be strongly connected, $n$ a positive integer, and $b$ be a real number with $b>0$.
 Then there exist real numbers $c,d>0$ such that for all real numbers $\epsilon>0$:
$$ 
\lim_{n\rightarrow\infty} \mu_{p}\left( E_{n} \left(\frac{c-\epsilon}{d}\right)\right)=0 
$$
\end{lemma}

\begin{proof}


Now, the DFA $A$ induces a stochastic
$\vert Q \vert \times \vert Q \vert$ matrix $\mathbf{P}$ 
by setting
\[\mathbf{P}_{ij} = \sum_{a \in \{0,1\}} 
\mu_p(a) \cdot [\delta(i,a)=j].\]
 Note in particular that
$\mathbf{P}_{ij} =0$ if{f} there are no transitions from
$i$ to $j$ in $Q$ on a symbol $a \in \{0,1\}$
with $\mu_p(a) > 0$. 
As $A$ is strongly connected,
there exists a path from state $i$ to state $j$ for each  $i,j \in Q$, and as $p$ is a positive
Bernoulli distribution, we have $\mu_p(a) = p(a) > 0$,$i,j$, whence for each $i,j$ there is an integer
$n_{ij}$ such that $\mathbf{P}^{n^{ij}}_{ij} > 0$,
that is, $\mathbf{P}$ (and its associated Markov chains) is irreducible. As all states of a finite Markov chain
with irreducible transition matrix are positive recurrent, standard results (see, e.g., \cite[Thm.\ 54]{Serfozo:basics}) yield that there is a unique positive stationary distribution
$\pi : Q \longrightarrow [0,1]$
(s.t., for all $i \in Q$, $\pi(i) > 0$ 
and $\lambda(i) = \sum_{j \in Q} \lambda(j)\mathbf{P}_{ij}$). Furthermore,
 the expected return time $M_i$
to state $i$ satisfies $M_i = 1/\pi(i)$
\cite[Thm.\ 54]{Serfozo:basics}.

Let $D$ be the least common multiple of the set
$\{n_{ij} : (i,q) \in Q^2\}$, and
let $(X_n)_{n \geq 1} = (X_1,X_2,\ldots)$ be a Markov
chain with transition matrix $\mathbf{P}$ and some
initial distribution $\lambda$ on the states.

Consider, for each $i \in Q$, the variable $V(i)$ where 
$V_i$, where
$$
V_i(n) = \sum_{k=0}^{n-1} 1_{X_k = i}
$$
As $\mathbf{P}$ is irreducible, the Ergodic Theorem
for  Markov chains (see, e.g., \cite[Thm.\ 75]{Serfozo:basics}) yields that 
\begin{equation}\label{eq:YouSeeBigGirl}
\lim_{n\rightarrow\infty} \textrm{Pr}\left( \left\vert \frac{V_i(n)}{n} - \pi(i) \right\vert \geq \epsilon \right) =
\lim_{n\rightarrow\infty} \textrm{Pr}\left( \left\vert \frac{V_i(n)}{n} - \frac{1}{M_i} \right\vert \geq \epsilon \right) = 0
\end{equation}

Let $\alpha \in \{0,1\}^n$
and let $q_{A_j}(\alpha) = q_1 \cdots q_{n-1}$ be the
sequence of states visited when $A$ is given
$\alpha$ as input starting from state $j$
(i.e., $q_1 = j$). Observe that the probability of observing the state
sequence $q_1 \cdots q_{n-1}$ in a Markov
chain with transition matrix $\mathbf{P}$
is $\textrm{Pr}(q_1 \cdots q_{n-1}) = \mu_p(\{\alpha : q_{A_j}(\alpha) = q_1 \cdots q_{n-1})$.
and thus:
\begin{align}
\textrm{Pr}\left( q_1 \cdots q_{n-1} :
\left\vert \frac{\sum_{k=0}^{n-1} [q_k = i]}{n} - \pi(i)\right\vert \geq \epsilon\right)
&= \\
\mu_p\left(\alpha : q_{A_j}(\alpha) = q_1 \cdots q_n \land \left\vert \frac{\sum_{k=0}^{n-1} [q_k = i]}{n} - \pi(i)\right\vert \geq \epsilon\ \right)
&=\\
\mu_p\left(\alpha : q_{A_j}(\alpha) = q_1 \cdots q_n \land \left\vert \frac{V_i(n)}{n} - \pi(i)\right\vert \geq \epsilon\ \right)
\end{align}
Hence, by (\Cref{eq:YouSeeBigGirl}) and the above, we have
\begin{equation}\label{eq:limes_final}
\lim_{n\rightarrow \infty} \mu_p\left(\alpha : q_{A_j}(\alpha) = q_1 \cdots q_n \land \left\vert \frac{V_i(n)}{n} - \pi(i)\right\vert \geq \epsilon\ \right) = 0
\end{equation}
If $q_{A_j}({w}) = q_1 \cdots q_{n-1}$ and one of the states $q_i \in \{q_1,\ldots,q_{n-1}\}$ is an element
of $F$, then $A_j$ picks out $w_i$. Set
$c = \min_{q_j \in F} \pi(i)$. Then, for all
$j \in Q$, (\Cref{eq:limes_final}) yields
that $\lim_{n\rightarrow\infty}\mu_p(E_n(c-\epsilon)) = 0$.
The result now follows from
$\mu_p(E_n(c-\epsilon)) \leq \sum_{j\in Q} \mu_p(E_n(c-\epsilon),j)$.

\end{proof}

\begin{remark}
In Lemma \Cref{main:lemma1}, the assumption that
the DFA $A$ is strongly connected can be omitted
if we make the assumption that
every cycle of $A$ contains an accepting state.

Let $k$ be the maximal number of non-accepting states in any path in $A$
from one accepting state to another that
does not contain any other accepting states
than the start and end states of the path. As every cycle of $A$ contains an accepting state,
$k$ is well-defined. If ${w} = w_1 w_2 \cdots \in \{0,1\}^\omega$ and $A[{w}]$ is infinite,
then, by construction, $\vert A_q[w_1\cdots w_n] \vert \geq d$ where $n = d(k+1) + r$ and $0\leq r < k+1$. As $d = (n-r)/(k+1) > n/(k+1) - 1 
\geq n/(k+2)$ for $n > 2(k+1)$, 
we have $\vert A_q[w_1\cdots w_n] \vert \geq n/(k+2)$. Hence, for $n > 2(k+1)$,
$A_q[w_1\cdots w_n] > n/(k+2)$,
whence $E_n(1/(k+2),q) = \emptyset$ for
$n > 2(k+1)$, and thus  $\mu_p(E_n(1/(k+2),q)) = 0$;
setting $c=1$ and $d=1/(k+2)$ then proves the lemma).

The assumption that every cycle of $A$
contains an accepting state is occasionally made in the modern literature on Agafonov's Theorem, e.g.\ \cite{BECHER2013109}. The reason for not making this assumption is that it is unnecessary for strongly connected automata 
\end{remark}

\begin{lemma}\label{main:lemma3}
Let $S$ be a strategy, and let $F$ be finite subset of $\{0,1\}^{\ast}$. Let $R = F \setminus \{{w} : \exists {u} \in F . {u} \prec {w}\}$ be the set obtained from $F$ by removing words ${w}$  that already have a proper prefix in  $F$. Define, for each positive integer $n$, the set $M_n = \{{w}\in\{0,1\}^{n} : S({w})\in F\}$. Then, $\mu_{p}(M_n)\leqslant \mu_{p}(R)$.
\end{lemma}

\begin{proof}
Observe that 
$$
M_n = \bigcup_{{u} \in R} \{{w} : S({w}) \in F \land {u} \preceq S({w})\}
$$
and thus 
$$
\mu_p(M_n) = \sum_{{u} \in R} \mu_p(\{{w} : S({w}) \in F \land {u} \preceq S({w})\})
\leq \sum_{{u}\in R} \mu_p(\{{w} : {u} \preceq S({w})\})
$$
Thus, if, for any word ${u} = a_1 a_2 \cdots a_k$,
the set $M_{u}=\{{w}\in \{0,1\}^{n} : {u}\preceq S({w})\}$ satisfies $\mu_{p}(M_{u})\leqslant \mu_{p}({u})$, it follows that 
$$
\mu_p(M_n) \leq \sum_{{u} \in R}
\mu_p(\{{w} : {u} \preceq S({w})\})
\leq \sum_{{u} \in R} \mu_{p}({u}) = \mu_p(R)
$$
as desired.
We thus proceed to prove $\mu_{p}(M_{u})\leqslant \mu_{p}({u})$ by induction on $k = \vert {u} \vert$.
\begin{itemize}
\item Base case: ${u} = a \in \{0,1\}$, so $\mu_p({u}) = \mu_p(a) = p(a)$. Let $\alpha=x_{1}x_{2}\dots x_{n}$ be a word in $M_{u}$ and let $x_{f} \in \{0,1\}$ be the first symbol selected by $S$ when applied to $\alpha$; as $
\alpha \in M_{u}$, we have $x_f = a$. 
Now, for each $b \in \{0,1\} \setminus\{x_f\}$,
define
$\bar{\alpha}_b=x_{1}x_{2}\cdots x_{f-1}b x_{f+1}\dots x_{n}$, that is, $\bar{\alpha}_b$ is the word obtained from $\alpha$ by changing the $f$th symbol to $b$.
Then, $\bar{\alpha}_b \notin M_{u}$. 

We define the set $\bar{M}_{u} =\{\bar{\alpha}_b : \alpha\in M_{u},
b \in \{0,1\} \setminus \{a\}\}$. Observe that
$\mu_{p}(\bar{\alpha}_b) = \mu_p(\alpha)p(b)/p(a)$,
and hence:
\begin{align*}
\mu_p(\bar{M}_{u}) &= \sum_{\alpha \in M_{u}}\sum_{b \in \{0,1\} \setminus \{a\}} \mu_p(\alpha)\frac{p(b)}{p(a)}
= \sum_{\alpha \in M_{u}} \frac{\mu_p(\alpha)}{p(a)}\left( \sum_{b \in \{0,1\} \setminus \{a\}} p(b) \right)\\
&= \frac{1-p(a)}{p(a)}
\sum_{\alpha\in M_{u}} \mu_p(\alpha)
= \frac{1-p(a)}{p(a)}
\mu_p(M_{u})
\end{align*}

Furthermore,
as $\bar{\alpha}_b \notin M_{u}$ for any $b \in \{0,1\} \setminus \{a\}$,
we have $M_{u} \cap \bar{M}_{u} = \emptyset$,
whence 
$\mu_p(M_{u}) + \mu_p(\bar{M}_{u}) \leq \mu_p(\{0,1\}^n) = 1$
and therefore $\mu_p(M_{u}) \leq 1 - \mu_p(\bar{M}_{u})$. Thus,
$$
\mu_p(M_{u}) \leq 1 - \mu_p(M_{u})\frac{1 - p(a)}{p(a)}
$$
that is,
$$
\mu_p(M_{u}) \leq \frac{1}{1 + \frac{1-p(a)}{p(a)}} = p(a) = \mu_p({u})
$$
as desired.

\item Inductive case: ${u}=a_{1}a_{2}\dots a_{k}a_{k+1}$ with $a_{k+1}=a$ for some $a \in \{0,1\}$. We have $M_{u} =\{\alpha\in\{0,1\}^* : a_{1}\cdots a_{k}a\preceq S(\alpha)\}$. Given $\alpha\in M_{u}$, let for each $b \in \{0,1\} \setminus \{a\}$, $\bar{\alpha}_b$ be the word obtained from $\alpha$ by changing the $k+1$th symbol selected by $S$ to $b$. Observe that
$\neg (\bar{\alpha}_b \preceq {u})$. Define $\bar{M}_{u}$ to be the set $\{\bar{\alpha}_b : \alpha\in M, b \in \{0,1\} \setminus \{a\}\}$,
and note that $M_{u} \cap \bar{M}_{u} = \emptyset$, and that 
$\mu_p(\bar{\alpha}_b) = \mu_p(\alpha)p(b)/p(a)$ and thus, as above,
$\mu_p(\bar{M}_{u}) = \mu_p(M_{u})(1-p(a))/p(a)$.

Let $N_{u}$ be the set $\{\alpha\in\{0,1\}^* :  a_{1}\dots a_{k}\preceq S(\alpha)\}$. Then $N_{u}$ contains as subsets both $M_{u}$ and $\bar{M}_{u}$,
whence $\mu_p(M_{u}) + \mu_p(\bar{M}_{u}) \leq \mu_p(N_{u})$. The induction hypothesis furnishes that $\mu_{p}(N_{u})\leq \mu_{p}(a_{1}a_{2}\dots a_{k})$, and thus $\mu_{p}(M_{u})+\mu_{p}(\bar{M}_{u})\leq  \mu_{p}(a_{1}a_{2}\dots a_{k})$. As $\mu_p(\bar{M}_{u}) = \mu_p(M_{u})(1-p(a))/p(a)$, we deduce that
\begin{align*}
 \mu_{p}(M_{u}) \leq  \mu_{p}(a_{1}a_{2}\dots a_{k}) - \mu_p(\bar{M}_{u}) 
 = \mu_{p}(a_{1}a_{2}\dots a_{k}) - \mu_p(M_{u})\frac{1-p(a)}{p(a)} \\
\end{align*}
and thus that 
$$
\mu_p(M_{u}) \leq \frac{ \mu_{p}(a_{1}a_{2}\dots a_{k})}{1+\frac{1-p(a)}{p(a)}} = \mu_{p}(a_{1}a_{2}\dots a_{k}) p(a) = \mu_p(a_{1}a_{2}\dots a_{k}a)
$$
as desired.

\end{itemize}

\end{proof}

\begin{lemma}\label{main:lemma2}
Let $S$ be a strategy, $a \in \{0,1\}$, $b, \epsilon$ be real numbers with
$0 < b \leq 1$ and $\epsilon > 0$, and define, for all positive integers $n$:
\begin{align*}
H_{n}(b,\epsilon) &=\left\{{w}\in\{0,1\}^{n} : \vert S({w})\vert >bn \land \left\vert p(a) - \frac{\countones{a}{S({w})}}{\vert S({w}) \vert} \right\vert \geq \epsilon\right\} \\
&= \bigcup_{bn < \ell \leq n} \left\{{w} \in \{0,1\}^n : S({w}) \in \{0,1\}^\ell \land  \left\vert p(a) - \frac{\countones{a}{S({w})}}{\ell} \right\vert \geq \epsilon \right\}
\end{align*}
Then:
$$
\lim_{n\rightarrow\infty}\mu_{p}(H_{n}(b,\epsilon))=0
$$
\end{lemma}

\begin{proof}
Define
$$
F_n(b,\epsilon) = \bigcup_{bn < \ell \leq n} \left\{ {y} \in \{0,1\}^{\ell} : \left\vert p(a) - \frac{\countones{a}{{y}}}{\ell} \right\vert \geq \epsilon \right\}
$$
Observe that 
$H_n(b,\epsilon) = \left\{{w} \in \{0,1\}^n : S({w}) \in F_n(b,\epsilon) \right\}$. Let $R_n(b,\epsilon) \subseteq \{0,1\}^{\leq n}$ be the set obtained by removing from 
$F_n(b,\epsilon)$ all ${w}$
such that there is ${u} \in F_n(b,\epsilon)$ with ${u} \prec {w}$ (i.e., remove all words from $F_n(b,\epsilon)$ that already have a prefix in $F_n(b,\epsilon)$).
 \Cref{main:lemma3} yields that
$\mu_p(H_n(b,\epsilon)) \leq \mu_p(R_n(b,\epsilon))$, 
and thus that $\lim_{n\rightarrow\infty} \mu_p(R_n(b,\epsilon)) = 0$.

Consider the stochastic variable $X$
that is $1$ when $1$ is picked from $\{0,1\}$
with probability $p$, and $0$ otherwise. Then, the mean of $X$ is $p$ and
the variance of $X$ is $p(1-p)$. Now consider
performing $\ell \geq 1$ independent Bernoulli trials
drawn according to $X$. Define $q(1) = p(a)$, $q(0) = 1-p(a)$, and
$q(1c) = p(a)q(c)$ and $q(0c) = (1-p(a))q(c)$
for $c \in \{0,1\}^+$, and consider the probability
distribution $\bar{q} : \{0,1\}^\ell \longrightarrow [0;1]$ on $\{0,1\}^\ell$. Now, for any 
${v} \in \{0,1\}^\ell$, the probability of obtaining ${v}$ by performing $\ell$ repeated Bernoulli trials as above
is $p^{\countones{{v}}}(1-p)^{\ell - \countones{{v}}} = \mu_p({v})$,
and hence for any event $\mathcal{U} \subseteq \{0,1\}^\ell$, we have for $X^\ell$:
\begin{equation}\label{eq:simple_mu}
\textrm{Pr}\left(\mathcal{U}\right) = \sum_{{u} \in \mathcal{U}} \mu_p({u}) = \mu_p\left(\mathcal{U}\right)
\end{equation}

Define the stochastic
variable $X^{\ell} = X + X + \cdots +  X$ ($\ell$ times). Then, $X^\ell$ counts the number of occurrences of $1$ by performing 
 $\ell$ Bernoulli trials as above.
By Chebyshev's inequality, $X^{\ell}$ satisfies:
\begin{equation}\label{main:eq:Cheb}
\textrm{Pr}\left(\left\vert p -  \frac{X^{\ell}}{\ell} \right\vert \geq \epsilon \right) \leq \frac{p^2(1-p)^2}{\ell \epsilon^2}
\end{equation}

The event $\vert p - X^{\ell}/\ell  \vert \geq \epsilon$
 is shorthand for the set
 \begin{align*}
  \left\{ {u} \in \{0,1\}^\ell : \left\vert p - \frac{\sum_{j=1}^\ell u_j}{\ell} \right\vert \geq \epsilon\right\} &=
\left\{ {u} \in \{0,1\}^\ell : \left\vert p - \frac{\countones{{u}}}{\ell} \right\vert \geq \epsilon\right\} 
\end{align*}

By (\ref{eq:simple_mu} we thus have:
\begin{align}
\textrm{Pr}\left( \left\vert p - \frac{X^{\ell}_a}{\ell}  \right\vert \geq \epsilon \right) &= \textrm{Pr}\left( \left\{ {u} \in \{0,1\}^\ell : \left\vert p - \frac{\countones{{u}}}{\ell} \right\vert \geq \epsilon\right\} \right) \nonumber\\
&= \mu_p\left(  \left\{ {u} \in \{0,1\}^\ell : \left\vert p - \frac{\countones{{u}}}{\ell} \right\vert \geq \epsilon\right\} \right)\label{it_foo}
\end{align}

Observe that: 
\begin{align}
\mu_p\left(R_n(b,\epsilon) \right) &= \mu_p \left(\bigcup_{bn < \ell \leq n} \left\{ {u} \in \{0,1\}^{\ell} \cap R_n(b,\epsilon) : \left\vert p - \frac{\countones{{u}}}{\ell} \right\vert \geq \epsilon \right\} \right) \nonumber \\
&= \sum_{bn < \ell \leq n} \mu_p\left(  \left\{ {u} \in \{0,1\}^{\ell} \cap R_n(b,\epsilon) : \left\vert p - \frac{\countones{{u}}}{\ell} \right\vert \geq \epsilon \right\}  \right) \label{eq:stop_having_fun}
\end{align}
But as $\mu_p(a_1 \cdots a_{\ell}) \geq \mu_p(a_1 \cdots a_\ell a_{\ell+1})$
for any $a_1,\ldots,a_\ell,a_{\ell+1} \in \{0,1\}$ and no element of $R_n(b,\epsilon)$ 
is a prefix of any other element, we have
\begin{eqnarray}
\lefteqn{\sum_{bn < \ell \leq n} \mu_p\left(  \left\{ {u} \in \{0,1\}^{\ell} \cap R_n(b,\epsilon) : \left\vert p - \frac{\countones{{u}}}{\ell} \right\vert \geq \epsilon \right\}  \right)} \nonumber \\
&\leq &
\mu_p\left(  \left\{ {u} \in \{0,1\}^{\lfloor bn \rfloor}: \left\vert p - \frac{\countones{{u}}}{\lfloor bn \rfloor} \right\vert \geq \epsilon \right\}  \right) \label{eq:get_it_over_with}
\end{eqnarray}

We thus have:

\begin{align*}
\mu_p(R_n(b,\epsilon) &\leq  \mu_p\left(  \left\{ {u} \in \{0,1\}^{\lfloor bn \rfloor}: \left\vert p - \frac{\countones{{u}}}{\ell} \right\vert \geq \epsilon \right\}  \right) & \textrm{by } (\cref{eq:stop_having_fun}) \textrm{ and } (\cref{eq:get_it_over_with}) \\
&= \textrm{Pr}\left( \left\vert p - \frac{X^{\lfloor bn \rfloor}}{\lfloor bn \rfloor} \right\vert \geq \epsilon \right) & \textrm{by } (\Cref{it_foo}) \\
&\leq \frac{p^2(1- p)^2}{\lfloor bn \rfloor \epsilon^2} & \textrm{by } (\Cref{main:eq:Cheb})
\end{align*}
Thus, 
$\lim_{n \rightarrow \infty} \mu_p(R_n(b,\epsilon)) = 0$, as desired.
\end{proof}

\begin{corollary}\label{cor:G_tends}
Let $b,\epsilon$ be real numbers
with $0 < b \leq 1$ and $\epsilon > 0$.
Then,
$$
\lim_{n\rightarrow\infty}\mu_{p}(G_{n}(b,\epsilon))=0
$$
\end{corollary}

\begin{proof}
By Lemma \Cref{main:lemma2} with
$S = A_q$, we obtain $\lim_{n\rightarrow\infty} \mu_p(G_n(b,\epsilon,q)) = 0$
and as $G_n(b,\epsilon) = \bigcup_{q\in Q}G_{n}(b,\epsilon,q)$,
we have
$\mu_p(G_n(b,\epsilon)) \leq \sum_{q\in Q}\mu_p(G_n(b,\epsilon,q))$.
As $Q$ is finite, we hence obtain
$\lim_{n \rightarrow\infty} \mu_p(G_n(b,\epsilon)) = 0$.
\end{proof}

\begin{lemma}\label{lem:mainclaim}
There is a real number $b$ with $0 < b \leq 1$ such that for all $\epsilon > 0$,
\[ \lim_{n\rightarrow\infty} \mu_{p}(D_{n}^{p}(b,\epsilon)) = 1. \]
\end{lemma}

\begin{proof}
Observe that, for all $b$ with $0 < b \leq 1$:
\begin{align*}
\{0,1\}^n \setminus D_{n}^{p}(b,\epsilon)) &= \left\{{w} \in \{0,1\}^n :
\exists q \in Q . \vert A_q[{w}] \vert \leq bn\right\}
\\ 
&\cup \left\{ {w} \in \{0,1\}^n :  \exists q \in Q . \vert
 A_q[{w}] a\vert > bn \land \max_{a \in \{0,1\}} \left\vert \frac{\countones{a}{A_q[{w}]}}{\vert A_q[{w}] \vert}  - p \right\vert \geq \epsilon \right\} \\
 &= \left( \bigcup_{q \in Q} E_n(b,q) \right) \cup \left(\bigcup_{q \in Q} G_n(b,\epsilon,q) \right)
\end{align*}
and thus, 
\begin{align*}
\mu_p(\{0,1\}^n \setminus D_n^{p}(b,\epsilon)) &\leq
\mu_p\left(  \bigcup_{q \in Q} E_n(b,q)  \right) + \mu_p\left( \bigcup_{q \in Q} G_n(b,\epsilon,q)  \right)\\
&= \mu_p(G_n(b,\epsilon)) + \mu_p(E_n(b))
\end{align*}
Choose, by \Cref{main:lemma1} real numbers $c,d$ such that
$\lim_{n\rightarrow\infty} \mu_{p}(E_{n}(\frac{c-\epsilon}{d}))=0$,
and set $b = (c - \epsilon)/d$.
By Corollary \Cref{cor:G_tends}, we obtain
$\lim_{n\rightarrow \infty} G_n(b,\epsilon) = 0$,
and thus 
$\lim_{n\rightarrow\infty}\mu_p(\{0,1\}^n \setminus D_n^{p}(b,\epsilon)) = 0$.
The result now follows by $\mu_p(D_n^{p}(b,\epsilon))) = 1 - \mu_p(\{0,1\}^n \setminus D_n^{p}(b,\epsilon))$.
\end{proof}

\begin{theorem}
Let $\alpha$ be a $p$-block-distributed right-infinite sequence, and $A$ a strongly connected DFA. Then the sequence $\beta=\pickedout{\alpha}{A}$ is $p$-distributed.
\end{theorem}

\begin{proof}
By \Cref{lem:simply_normal_is_enough} it suffices to show
that, for all $a \in \{0,1\}$, the limiting frequency of $a$ in $A[\alpha]$ 
exists and is equal to $p$.

Consider the sequence $(\beta_{(n,r)})$ of blocks of $\pickedout{x}{A}$ corresponding to the sequence of blocks $(\alpha_{(n,r)})$, that is 
$\beta_{(n,r)}$ is the sequence of symbols picked out from
block $\alpha_{(n,r)}$ when $A$ is applied to $\alpha$;
note that
each $\beta_{[n,r]}$ has  length between $0$ and $n$. 

For each positive integer $m$, define
$L_m =\sum_{i=1}^{m}\vert \beta_{[n,i]} \vert$, and
 for each $a \in \{0,1\}$, write $\rho_a^m =\frac{\sum_{i=1}^{m} \countones{a}{\beta_{(n,i)}}}{L}$. 
 Observe that, to prove the theorem, it suffices to show
 that, for any real number $\epsilon$ with $ 0 < \epsilon < 1$ and sufficiently large
 $m$, that $\abs{\rho_a-p}<\epsilon$.
 
Furthermore, set
$I_m=\{i\leqslant m : \alpha_{(n,i)}\not\in D_{n}^{p}(b,\frac{\epsilon}{2})\}$, and set
$\ell_m =\sum_{i\in I_m}\vert \beta_{(n,i)} \vert$.

Now, define $\theta_a^m$ by:
$$
\theta_a^m = \frac{\sum_{i \in \{1,\ldots,m\} \setminus I_m} \countones{a}{\beta_{(n,i)}}}{\sum_{i \in \{1,\ldots,m\} \setminus I_m} \vert {y}_{(n,i)}\vert} = \frac{\sum_{i \in \{1,\ldots,m\} \setminus I_m} \countones{a}{\beta_{(n,i)}}}{L_m - \ell_m}
$$
That is, $\theta_a^m$ is the frequency of occurrences of $a$s when the blocks $\beta_{(n,i)]}$ picked out from blocks $\alpha_{(n,r)} \in D_{n}^{p}(b,\frac{\epsilon}{2})$ are concatenated.
Observe that, by definition
of $D^p_n$, we have $\abs{\theta^m_a-p}<\frac{\epsilon}{2}$.

We have:
 \begin{align}
 \rho_a^m - \theta_a^m &= 
 \frac{\sum_{i=1}^{m} \countones{a}{\beta_{(n,i)}}}{L_m} -
 \frac{\sum_{i\in \{1,\ldots,m\} \setminus I} \countones{a}{\beta_{(n,i)}}}{L_m - \ell_m} \nonumber\\
&=
 \left(\frac{\sum_{i\in I_m} \countones{a}{\beta_{(n,i)}}}{L_m} + \frac{\sum_{i \in \{1,\ldots,m\} \setminus I_m}
 \countones{a}{\beta_{(n,i)}}}{L_m}\right) - \frac{\sum_{i \in \{1,\ldots,m\} \setminus I}\countones{a}{\beta_{(n,i)}}}{L_m-\ell_m}
\nonumber \\
 &\overset{(\dagger)}{=}
 \frac{\sum_{i \in \{1,\ldots,m\} \setminus I_m}
 \countones{a}{\beta_{(n,i)}}}{L_m} - \frac{\sum_{i \in \{1,\ldots,m\} \setminus I_m}\countones{a}{\beta_{(n,i)}}}{L_m-\ell_m}
 +
  \frac{\sum_{i\in I_m} \countones{a}{\beta_{(n,i)}}}{L_m} \nonumber\\ 
  &\leq \frac{\sum_{i\in I_m} \countones{a}{\beta_{(n,i)}}}{L_m}
  \leq \frac{\sum_{i\in I_m} \vert \beta_{(n,i)}\vert}{L_m} = \frac{\ell_m}{L_m}
  \label{final_line_foo}
  \end{align}
\noindent where the penultimate inequalities in the last line above
follows because $L_m \geq L_m - \ell_m$ implies
$\frac{\sum_{i \in \{1,\ldots,m\} \setminus I}
 \countones{a}{\beta_{(n,i)}}}{L} - \frac{\sum_{i \in \{1,\ldots,m\} \setminus I}\countones{a}{\beta_{[n,i]}}}{L-\ell} \leq 0$,
 and the final inequality follows because
 $\sum_{i\in I} \countones{a}{\beta_{[n,i]}}\leq \sum_{i\in I}\vert \beta_{[n,i]} \vert = \ell_m$.
 
By basic algebra, we have: 
$$
\frac{\sum_{i \in \{1,\ldots,m\} \setminus I_m}
 \countones{a}{\beta_{(n,i)}}}{L_m} - \frac{\sum_{i \in \{1,\ldots,m\} \setminus I_m}\countones{a}{\beta_{(n,i)}}}{L_m-\ell_m}
 =
 \frac{-\ell_m \sum_{i\in\{1,\ldots,m\} \setminus I}\countones{a}{\beta_{(n,i)}}}{L_m(L_m-\ell_m)}
$$
and as
$$
\sum_{i\in\{1,\ldots,m\} \setminus I_m}\countones{a}{\beta_{(n,i)}} \leq 
\sum_{i \in \{1,\ldots,m\} \setminus I_m}
\leq \sum_{i\in\{1,\ldots,m\} \setminus I_m} \vert \beta_{(n,i)} \vert \leq
L_m-\ell_m
$$
we conclude that 
$$
\frac{-\ell \sum_{i\in\{1,\ldots,m\} \setminus I}\countones{a}{\beta_{[n,i]}}}{L(L-\ell)} \geq -\frac{\ell_m}{L_m}
$$
and thus by ($\dagger$) that
$$
\rho_a^m - \theta_a^m + \frac{\sum_{i\in I_m} \countones{a}{\beta_{(n,i)}}}{L_m} \geq - \frac{\ell_m}{L_m}
$$
\noindent whence $-\ell_m/L_m \leq \rho_a - \theta_a$,
which combined with (\Cref{final_line_foo}) yields
$\vert \rho_a - \theta_a \vert \leq \ell/L$.

By  \Cref{lem:mainclaim} pick a $b$ such that such that for all $\epsilon > 0$, we have $\lim_{n\rightarrow\infty} \mu_{p}(D_{n}^{p}(b,\epsilon)) = 1$. Choose $\delta>0$ with $\delta<\frac{b\epsilon}{8}$. Pick $n\in\naturalN$ such that $\mu_{p}(D_{n}^{p}(b,\epsilon))>1-\delta$. Now, pick $\alpha<\frac{b\epsilon}{8}$. Because $\alpha$ is $p$-block-distributed, there exists $M \in\naturalN$ such that for all $k \geq M$ and
all $\seq{G}\subseteq \{0,1\}^{n}$, the prefix
$\alpha_{\leq kn}$ of $\alpha$ of length $kn$ satisfies:
$$
\left\vert 
\frac{\vert \{ i \leq k : \alpha_{(n,i)} \in \seq{G}\} \vert}{k} - \mu_p(\seq{G})
\right\vert < \alpha
$$
In the particular case ${G} = D^p_n(b,\epsilon/2)$, we thus have:
$$
\left\vert \frac{\vert \{ i \leq k : \alpha_{(n,i)} \in D_n^p\left(b,\frac{\epsilon}{2}\right)\} \vert}{k} - \mu_p\left( D^p_n\left(b,\frac{\epsilon}{2}\right) \right) 
\right\vert < \alpha
$$
and thus
{\small
\begin{align*}
1 - \delta  - \frac{\vert \{ i \leq k : \alpha_{(n,i)} \in D_n^p(b, \frac{\epsilon}{2})\} \vert}{k} &\leq 
\mu_p\left( D^p_n\left( b,\frac{\epsilon}{2} \right) \right)
- \frac{\vert \{ i \leq k : \alpha_{(n,i)} \in D_n^p(b,\frac{\epsilon}{2})\} \vert}{k} \\ 
& < \alpha
\end{align*}
}
and thus 
\begin{equation}\label{eq:almost_final}
\left\vert \left\{ i \leq k : \alpha_{(n,i)} \in D_n^p\left( b,\frac{\epsilon}{2}\right)\right\} \right\vert > k(1-\delta - \alpha)
\end{equation}

By definition of $D_{n}^{p}(b,\frac{\epsilon}{2}))$, every
$\alpha_{(n,i)} \in D_{n}^{p}(b,\frac{\epsilon}{2}))$
satisfies $\vert \pickedout{\alpha_{(n,i)}}{A} \vert > bn$,
and we thus have, whence 
$$
L_m = \sum_{i=1}^m \vert {y}_{(n,i)} \vert
= \sum_{i=1}^m \vert  \pickedout{\alpha_{(n,i)}}{A} \vert
\geq \left\vert \left\{ i \leq m : \alpha_{(n,i)} \in D_n^p\left( b,\frac{\epsilon}{2}\right)\right\} \right\vert bn > m(1 - \delta - \alpha)bn
$$
Furthermore, by definition of $I_m$ and (\Cref{eq:almost_final}),
\begin{align*}
\vert I_m \vert &= \left\vert \left\{i\leqslant m : \alpha_{(n,i)}\not\in D_{n}^{p}\left(b,\frac{\epsilon}{2}\right)\right\}\right\vert = m - \left\vert \left\{ i \leq m : \alpha_{(n,i)} \in D_n^p\left( b,\frac{\epsilon}{2}\right)\right\} \right\vert\\
&< m - m(1-\delta-\alpha) = m(\delta+\alpha)
\end{align*}
But then,
$$
\ell_m = \sum_{i \in I_m} \vert {y}_{(i,n)} \vert \leq \vert I_m \vert n < mn(\delta + \alpha)
$$
\noindent and thus:
$$
\frac{\ell_m}{L_m} < \frac{mn(\delta + \alpha)}{m(1 - \delta -\alpha)bn} = \frac{\delta + \alpha}{b(1 - \delta - \alpha)} < \frac{\frac{b\epsilon}{8} + \frac{b\epsilon}{8}}{b\left( 1- \frac{b\epsilon}{8} - \frac{b\epsilon}{8}\right)}
< \frac{\frac{\epsilon}{8}}{1 - \frac{1}{4}} < \frac{\epsilon}{2}
$$
where we have used that $b\epsilon < 1$ in the penultimate inequality.

We now finally have
$$
\vert \rho_a -p \vert \leq \vert\rho^m_a - \theta^m_a \vert + \vert \theta_a -p \vert < \frac{\ell_m}{L_m} + \frac{\epsilon}{2}
< \frac{\epsilon}{2} + \frac{\epsilon}{2} = \epsilon
$$
\noindent concluding the proof.
\end{proof}

\bibliographystyle{abbrv}
\bibliography{agafonovbib}

\begin{thebibliography}{10}

\bibitem{Agafonovsummary}
V.~N. Agafonov.
\newblock {Normal sequences and finite automata}.
\newblock {\em {Sov. Math., Dokl.}}, 9:324--325, 1968.
\newblock Originally published in {R}ussian (vol.\ 179:2, p. 255-266).

\bibitem{AgafonovsummaryRussian}
V.~N. Agafonov.
\newblock Normal sequences and finite automata.
\newblock {\em Dokl. Akad. Nauk SSSR}, 179(2):255--256, 1968.

\bibitem{BECHER2013109}
V.~Becher and P.~A. Heiber.
\newblock Normal numbers and finite automata.
\newblock {\em Theoretical Computer Science}, 477:109--116, 2013.

\bibitem{Borel}
E.~Borel.
\newblock Les probabilit\'{e}s d\'{e}nombrables et leurs applications
  arithm\'{e}tiques.
\newblock {\em Rend. Circ. Matem. Palermo}, 27:247--271, 1909.

\bibitem{Church40}
A.~Church.
\newblock On the concept of a random sequence.
\newblock {\em Bulletin of the American Mathematical Society}, 46(2):130--135,
  1940.

\bibitem{Copeland28}
A.~H. Copeland.
\newblock Admissible numbers in the theory of probability.
\newblock {\em American Journal of Mathematics}, 50(4):535--552, 1928.

\bibitem{Copeland36}
A.~H. Copeland.
\newblock Point set theory applied to the random selection of the digits of an
  admissible number.
\newblock {\em American Journal of Mathematics}, 58(1):181--192, 1936.

\bibitem{Kamke33}
E.~Kamke.
\newblock {\"U}ber neuere begr{\"u}ndungen der {W}ahrscheinlichkeitsrechnung.
\newblock {\em Jahresbericht der Deutschen Mathematiker-Vereinigung},
  42:14--27, 1933.

\bibitem{Loveland66}
D.~W. Loveland.
\newblock The kleene hierarchy classification of recursively random sequences.
\newblock {\em Transactions of the American Mathematical Society},
  125(3):497--510, 1966.

\bibitem{Agafonov}
{\fontencoding{OT2}\selectfont В.
  Н.}.~{\fontencoding{OT2}\selectfontАгафонов}.
\newblock {\fontencoding{OT2}\selectfont Нормальные
  последовательности и конечные автоматы}.
\newblock {\em {\fontencoding{OT2}\selectfont Докл. АН СССР}},
  179(2):255--256, 1968.

\bibitem{AgafonovRussianLong}
{\fontencoding{OT2}\selectfont В.
  Н.}.~{\fontencoding{OT2}\selectfontАгафонов}.
\newblock {\em {\fontencoding{OT2}\selectfont Нормальные
  последовательности и конечные автоматы}},
  volume~20, pages 123--129.
\newblock {\fontencoding{OT2}\selectfont Наука, Академии наук
  СССР}, 1968.

\bibitem{Postnikov}
{\fontencoding{OT2}\selectfont А. Г.}.~{\fontencoding{OT2}\selectfont
  Постников}.
\newblock {\fontencoding{OT2}\selectfont Арифметическое
  моделирование случайных процессов}.
\newblock {\em {\fontencoding{OT2}\selectfont Тр. МИАН СССР}},
  57:3--84, 1960.

\bibitem{PostnikovPyateskii}
{\fontencoding{OT2}\selectfont А. Г.}.~{\fontencoding{OT2}\selectfont
  Постников} and {\fontencoding{OT2}\selectfont И.
  И.}.~{\fontencoding{OT2}\selectfont Пятецкий}.
\newblock {\fontencoding{OT2}\selectfont Нормальные по
  Бернулли последовательности знаков}.
\newblock {\em {\fontencoding{OT2}\selectfont Изв. АН СССР. Сер.
  матем.}}, 21(4):501--514, 1957.

\bibitem{Postnikova}
{\fontencoding{OT2}\selectfont Л. П.}.~{\fontencoding{OT2}\selectfont
  Постникова}.
\newblock {\fontencoding{OT2}\selectfont О связи понятий
  коллектива Мизеса--Черча и нормальной по
  Бернулли последовательности знаков}.
\newblock {\em {\fontencoding{OT2}\selectfont Теория вероятн. и
  ее примен.}}, 6(2):232--234, 1961.

\bibitem{Postnikova61}
L.~Postnikova.
\newblock On the connection between the concepts of collectives of
  {M}ises-{C}hurch and normal {B}ernoulli sequences of symbols.
\newblock {\em Theory of Probability \& Its Applications}, 6(2):211--213, 1961.
\newblock translation of \cite{Postnikova} by Eizo Nishiura.

\bibitem{Reichenbach32}
H.~Reichenbach.
\newblock Axiomatik der wahrscheinlichkeitsrechnung.
\newblock {\em Mathematische Zeitschrift}, 34(1):568--619, 1932.

\bibitem{Reichenbach37}
H.~Reichenbach.
\newblock Les fondements logiques du calcul des probabilit{\'e}s.
\newblock In {\em Annales de l'institut Henri Poincar{\'e}}, volume~7, pages
  267--348, 1937.

\bibitem{DBLP:journals/acta/SchnorrS72}
C.~Schnorr and H.~Stimm.
\newblock {E}ndliche {A}utomaten und {Z}ufallsfolgen.
\newblock {\em Acta Informatica}, 1:345--359, 1972.

\bibitem{Serfozo:basics}
R.~Serfozo.
\newblock {\em Basics of Applied Stochastic Processes}.
\newblock Probability and Its Applications. Springer-Verlag, 2009.

\bibitem{shields:bernoulli}
P.~Shields.
\newblock {\em The Theory of Bernoulli Shifts}.
\newblock Univ. Chicago Press, 1973.

\bibitem{Tornier29}
E.~Tornier.
\newblock Wahrscheinlichkeitsrechnung und {Z}ahlentheorie. erste {M}itteilung.
\newblock {\em Journal f{\"u}r die reine und angewandte Mathematik},
  1929(160):177--198, 1929.

\bibitem{vonMiseskollektiv}
R.~Von~Mises.
\newblock Grundlagen der {W}ahrscheinlichkeitsrechnung.
\newblock {\em Mathematische Zeitschrift}, 5(191):52--99, 1919.

\bibitem{vonMises36}
R.~von Mises.
\newblock {\em Wahrscheinlichkeit Statistik und Wahrheit}.
\newblock Springer, 1936.

\end{thebibliography}

\end{document}